\newif\ifdraft \drafttrue

\documentclass{article}
\usepackage{booktabs} % For formal tables
\usepackage[ruled]{algorithm2e} % For algorithms

\SetAlFnt{\small}
\SetAlCapFnt{\small}
\SetAlCapNameFnt{\small}
\SetAlCapHSkip{0pt}
\IncMargin{-\parindent}

\usepackage{natbib}
\usepackage{fullpage}
\usepackage{bbm}
\usepackage{xcolor}
\usepackage{authblk}
\usepackage{graphicx} % Allows including images
\usepackage{booktabs} % Allows the use of \toprule, \midrule and \bottomrule in tables
\usepackage{caption}
\usepackage{pgfplots}
\usepackage{float}
\pgfplotsset{compat=1.10}
\usepackage{subcaption}
\usepackage{tikz}
\usepgfplotslibrary{fillbetween}
\usetikzlibrary{patterns}

\usetikzlibrary{external}
\usepackage{amsmath,amsthm,amssymb,amsfonts}
\usepackage{accents}

\usepackage{thmtools}
\usepackage{thm-restate}

\def\cdf(#1)(#2)(#3){0.5*(1+(erf((#1-#2)/(#3*sqrt(2)))))}%

\usetikzlibrary{math}
\usepackage{pgfplots}
\pgfplotsset{compat=1.14}
\tikzmath{%
  function h1(\x, \lx) { return (9*\lx + 3*((\lx)^2) + ((\lx)^3)/3 + 9); };
  function h2(\x, \lx) { return (3*\lx - ((\lx)^3)/3 + 4); };
  function h3(\x, \lx) { return (9*\lx - 3*((\lx)^2) + ((\lx)^3)/3 + 7); };
  function skewnorm(\x, \l) {
    \x = (\l < 0) ? -\x : \x;
    \l = abs(\l);
    \e = exp(-(\x^2)/2);
    return (\l == 0) ? 1 / sqrt(2 * pi) * \e: (
      (\x < -3/\l) ? 0 : (
      (\x < -1/\l) ? \e / (8 * sqrt(2 * pi)) * h1(\x, \x*\l) : (
      (\x <  1/\l) ? \e / (4 * sqrt(2 * pi)) * h2(\x, \x*\l) : (
      (\x <  3/\l) ? \e / (8 * sqrt(2 * pi)) * h3(\x, \x*\l) : (
      sqrt(2/pi) * \e)))));
  };
}

\tikzset{
	declare function={
		normcdf(\x,\m,\s)=1/(1 + exp(-0.07056*((-\x-\m)/\s)^3 - 1.5976*(-\x-\m)/\s));
	}
}

\pgfmathdeclarefunction{gauss}{2}{%
	\pgfmathparse{1/(#2*sqrt(2*pi))*exp(-((x-#1)^2)/(2*#2^2))}%
}
\pgfmathdeclarefunction{gausseval}{3}{%
	\pgfmathparse{1/(#2*sqrt(2*pi))*exp(-((#3-#1)^2)/(2*#2^2))}%
}

\newcommand{\ind}{\mathbbm{1}}
\newcommand{\cj}[1]{\ifdraft \textcolor{red}{[Chris: #1]}\fi}
\newcommand{\ar}[1]{\ifdraft \textcolor{red}{[Aaron: #1]}\fi}

\newcommand{\cl}[1]{\ifdraft \textcolor{red}{[Changhwa: #1]}\fi}

\newcommand{\groups}{\mathcal{G}} % the group G = {1,2}
\newcommand{\group}{g} % group variable
\newcommand{\agent}{a} % agent
\newcommand{\crime}{c} % indicating whether an agent has committed crime
\newcommand{\guilty}{q} % indicating whether an agent is labeled guilty or not
\newcommand{\crimeReward}{\rho} % payoff of crime
\newcommand{\crimeCost}{\kappa} % cost of crime
\newcommand{\TPR}{\mathrm{TPR}} % true positive rate
\newcommand{\FPR}{\mathrm{FPR}} % false positive rate
\newcommand{\FNR}{\mathrm{FNR}} % false negative rate
\newcommand{\PPV}{\mathrm{PPV}} % positive predictive value
\newcommand{\outOptVal}{\omega} % outside option value
\newcommand{\outOptCdf}{H} % outside option cdf
\newcommand{\outOptPdf}{h} % outside option pdf
\newcommand{\sigVal}{s} % value of the signal
\newcommand{\fairnessNotion}{\xi} % talking about which fairness notion

\newcommand{\sigSet}{\mathbb{R}} % set of possible signals

\newcommand{\crimeSigCDF}{F^{\text{c}}} % signal CDF for criminals
\newcommand{\noncrimeSigCDF}{F^{\text{i}}} % signal CDF for non-criminals
\newcommand{\crimeSigPDF}{f^{\text{c}}} % signal CDF for criminals
\newcommand{\noncrimeSigPDF}{f^{\text{i} }} % signal CDF for non-criminals
\newcommand{\probPolicy}{\beta} % adjudicator's probabilistic policy
\newcommand{\margCostCrime}{\Delta} % disincentive of crime
\newcommand{\numPeople}{N} % number of people
\newcommand{\sigThresh}{T} % signal threshold
\newcommand{\inspecInsty}{\theta} % inspection intesnity
\newcommand{\searchCap}{S} % search capacity
\newcommand{\CFPR}{CFPR} % false positive rate conditioned on being inspected
\newcommand{\CTPR}{CTPR} % true negative rate conditioned on being inspected

\newcommand{\sigStruct}{F} % defines the signal structure (criminal signal distr, noncriminal signal distr)

\newcommand{\Beta}{B} % to denote a set of adjudicator's policy
\newcommand{\crimeRate}{\text{CR}} % to denote the fairness notion

\newtheorem{theorem}{Theorem}[section]
\newtheorem{remark}[theorem]{Remark}
\newtheorem{example}[theorem]{Example}
\newtheorem{definition}[theorem]{Definition}
\newtheorem{lemma}[theorem]{Lemma}

\begin{document}

\title{Fair Prediction with Endogenous Behavior}
%\author{Christopher Jung \and Sampath Kannan \and Changwa Lee \and Mallesh M. Pai \and Aaron Roth \and Rakesh Vohra}
%\affil{University of Pennsylvania}
\author[1]{Christopher Jung}
\author[1]{Sampath Kannan}
\author[2]{Changhwa Lee}
\author[3]{Mallesh M. Pai}
\author[1]{Aaron Roth}
\author[2]{Rakesh Vohra}
%\address{$^{\star}$Department of Computer Science, University of Pennsylvania}
\affil[1]{Department of Computer and Information Sciences, University of Pennsylvania}
\affil[2]{Department of Economics and Department of Electrical \& Systems Engineering, University of Pennsylvania}
\affil[3]{Department of Economics, Rice University}
% \address{$^{\dagger}$Department of Economics and Department of Electrical \& Systems Engineering, University of Pennsylvania}
% \address{$^{\ddag}$Department of Economics, Rice University}

% \thanks{Jung, Kannan, Lee, Roth, and Vohra gratefully acknowledge financial support from NSF grant CCF-1763307. Pai gratefully acknowledges financial support from NSF grant CCF-1763349.\\
% \emph{Please do not circulate.}}
\maketitle

\begin{abstract}
There is increasing regulatory interest in whether machine learning algorithms deployed in consequential domains (e.g. in criminal justice) treat different demographic groups ``fairly.'' However, there are several proposed notions of fairness, typically mutually incompatible.
Using criminal justice as an example, we study a model in which society chooses an incarceration rule. Agents of different demographic groups differ in their outside options (e.g. opportunity for legal employment) and decide whether to commit crimes. We show that equalizing type I and type II errors across groups is consistent with the goal of minimizing the overall crime rate; other popular notions of fairness are not.
\end{abstract}

%\newpage
%%% End of Disclosure Statements

\section{Introduction}

%\textbf{Undecided items from the meeting}
%\begin{itemize}
%    \item Signal corr with with outside option:
%     \item equality of outcome vs equality of opportunity? Yes, but %hesitated to use the term (moral philosophers' critiques)
%
%\end{itemize}

Algorithms to automate consequential decisions such as hiring \citep{hiring}, lending \citep{lending}, policing \citep{policing}, and criminal sentencing \citep{sentencing} are frequently suspected of being unfair or discriminatory. The suspicions are not hypothetical. The 2016 ProPublica study \citep{propublica} of the COMPAS Recidivism Algorithm (used to inform criminal sentencing decisions by attempting to predict recidivism) found that the algorithm was significantly more likely to incorrectly label black defendants as recidivism risks compared to white defendants, despite similar overall rates of prediction accuracy between populations. Since then, discoveries of ``algorithmic bias'' have proliferated, including a recent study of racial bias by algorithms that prioritize patients for healthcare \citep{sendhilScience}. Thus spurred, policymakers, regulators, and computer scientists have proposed that algorithms be designed to satisfy notions of fairness (see for instance \cite{mathbabe16,fairmlbook,ethicalalgorithm} for overviews).

This raises a question: what measure(s) of fairness should designers be held to, and how do these constraints interact with the original objectives the algorithm was designed to target? The COMPAS case illustrates that the answer is not clear. ProPublica and Northpointe (the company that designed COMPAS) advocated for different measures of fairness. ProPublica argued that the algorithm's predictions did not maintain parity in false positive and false negative rates between white and black defendants,%
\footnote{Northpointe's algorithm had differing  Type-1 and Type-2 error rates across the two groups.}
while Northpointe countered that their algorithm satisfied predictive parity.%
\footnote{Roughly, the accuracy of COMPAS scores was the same for both groups at all risk levels.}
Subsequent research identified hard trade-offs in the choice of fairness metrics: under some mild conditions, the two requirements above cannot simultaneously be satisfied  (\cite{KMR16}, \cite{Chou17}). This inspired a literature proposing (or criticizing) notions of fairness based on ethical/ normative grounds. The literature evaluates algorithms on the basis of these measures, and/or proposes novel algorithms that better trade-off the goals of the original designer (decision accuracy, algorithmic efficiency) with these fairness desiderata.%
\footnote{See e.g. \cite{DworkFairness,HPS16,sharad1,sharad2,sharad3,impossibility,gerrymander,multicalibration,implicit} for a small sample of an enormous literature.}
In general, the different proposed fairness measures are fundamentally at odds with one another. For example, in addition to the impossibility results due to \cite{KMR16,Chou17}, enforcing parity of false positive or false negative rates for e.g. parole decisions typically requires making parole decisions using different thresholds on the posterior probability that an individual will commit a crime for different groups. This has itself been identified by \citep{sharad1}  as a potential source of ``unfairness''.

This line of research is subject to two criticisms. First raised by, for example \citep{sharad1}: these notions of fairness are disconnected from and lead to unpalatable tradeoffs with other economic and social quantities and consequences one might care about. Second, the literature almost exclusively assumes that the agent types, which are relevant to the decision at hand, are exogenously determined, i.e. unaffected by the decision rule that is selected. For instance, in the criminal justice application described, individual choices of whether to commit a crime or not, and therefore the overall crime rates, are fixed and not affected by policy decisions made at a societal level (e.g. what legal standards are used to convict, policing decisions etc). In settings like this where agent decisions are exogenously fixed, \cite{sharad1} and \cite{implicit} observe that optimizing natural notions of welfare and accuracy (incarcerating the guilty, acquitting the innocent) are achieved by decision rules that select a uniform threshold on ``risk scores'' that are well calibrated --- for example, the posterior probability of criminal activity --- which tend \emph{not} to satisfy statistical notions of fairness that have been proposed in the literature. Does this mean that setting uniform thresholds on equally calibrated risk scores is better aligned with natural societal objectives than is asking for parity in terms of false positive and negative rates across populations?

In this paper, we consider a setting in which agent decisions are endogenously determined and show that in this model, the answer is \emph{no}: in fact, parity of false positive and negative rates (sometimes known in this literature as \emph{equalized odds} \cite{HPS16}) is aligned with the natural objective of minimizing crime rates. Parity of positive predictive value and posterior threshold uniformity are {\em not}. Although the model need not be tied to any particular application, we develop it using the language of criminal justice. We treat agents as rational actors whose decisions about whether or not to commit crimes are {\em endogenously} determined as a function of the incentives given by the decision procedure society uses to punish crime.  The possibility for unfairness arises because agents are ex-ante heterogeneous: their demographic group is correlated with their underlying incentives--- for example each individual has a private \emph{outside option} value for not committing a crime, and the distribution of outside options differs across groups. Our key result is that policies that are optimized to minimize crime rates are compatible with a popular measure of demographic fairness --- equalizing false positive and negative rates across demographics --- and are generally incompatible with equalizing positive predictive value and uniform posterior thresholds. Thus, which of these notions of fairness is compatible with natural objectives hinges crucially on whether one believes that criminal behavior is responsive to policy decisions or not.

Our results have direct implications for regulatory testing for unfairness. Often, in settings of interest, a regulator does not directly observe the decision rule used by an adjudicator. However, the regulator may wish to test whether the adjudicator is using a ``fair'' rule, i.e. whether the adjudicators choices are biased towards or against some demographic group. Following a tradition starting with \cite{becker}, one standard used is called an outcome test, i.e. comparing, ex-post, the classification assigned by the adjudicator to observed outcomes. For instance, in a criminal justice setting, one may compare the judge's decision to the (somehow obtained) actual innocence or guilt of the defendants, or in a lending setting, compare the lender's decision on whom to extend loans to with the actual repayment outcomes of loan applicants etc.

In this context, a given prescription on what constitutes a ``fair'' or non-discriminatory rule maps into a corresponding outcome test. In particular, a test that is popularly used by researchers and regulators  corresponds to the common-posterior-threshold rule described above. As already mentioned, this is not the best test in our model. When used, this test attempts to evaluate whether the adjudicator is using a common posterior threshold across groups by evaluating whether the marginal agents across groups have similar probabilities of different outcomes.
\footnote{For a discussion of this in the context of evaluating the fairness of lending standards, see \cite{ferguson1995constitutes}.}
However, implementing this test is difficult: identifying (and being sure that one has correctly identified) the marginal agent in each group is hard (this is roughly the infra-marginality problem, see e.g. \citep{simoiu2017problem}). For instance, there may be information observed by the decision maker but not by the regulator/ econometrician (an oft cited example is that police observe a suspect's demeanor, and use this as a factor, but this cannot be quantified).  By contrast, if our maintained assumptions are valid, then an adjudicator wishing to minimize crime should use a rule that equalizes false positive and false negative rates across demographic groups. This is easy to estimate and test: there is no need to identify a marginal agent.

\subsection{Overview of Model and Results}\label{sec:overview}
We first derive our results in an extremely simple baseline model to highlight the underlying intuition.  We then show that our conclusions are robust to a number of elaborations and generalizations of the model.

\subsubsection*{The Baseline Model}
Our baseline model (in Section \ref{sec:model}) has a mass of agents who each belong to one of two demographic groups. Each agent has a single choice on the extensive margin, for instance a binary choice of whether or not to commit a crime; or whether or not to acquire human capital, etc. To fix ideas, in this paper we frame the matter as a decision about whether to commit a crime.

An adjudicator has to classify each agent as guilty or innocent. This classification is based on  a noisy signal that the adjudicator receives of each agent's choice; the distribution of this signal depends only on the agent's choice, and not on her group. Further, the adjudicator observes the group membership of each agent. The adjudicator commits ex-ante to a classification rule, i.e. how it will classify agents as a function of the signal received, and potentially the agent's group membership.

Agents are expected payoff maximizers who enjoy a monetary benefit from crime but also a cost of being declared guilty of the crime. In choosing whether to commit a crime, they compare their expected net benefit from the crime to an outside option. The costs and benefits are privately known to the agent, but not to the adjudicator (who only sees group membership).  The only distinction between groups is that the distributions of costs and benefits may be different for different groups. For example, individuals from different groups might have  different legal employment opportunities, different costs of incarceration (e.g. differences in stigma), etc. The model is flexible enough to allow (potentially different) fractions of the population in each group who are rigidly law-abiding (i.e. do not commit a crime regardless of circumstance) or hardened criminals (i.e. will commit a crime regardless of circumstance), and a variety of responses to incentives in between these two extremes. We don't model the source of this heterogeneity: it is exogenous, and the distribution is known to the adjudicator. Given these preferences, the adjudicator's decision rule determines their choices, which in turn determines the overall crime rate in each population.

The adjudicator's objective is to minimize the overall crime rate, i.e. the total mass of agents that choose to commit a crime. While we model the adjudicator as knowing that the underlying groups are heterogeneous (i.e. knowing the above distributions that describe each group), the adjudicator is not biased for or against any group, nor is there any underlying preference for fairness. Our main result (see Section \ref{sec:baseline}) is that the classifier that minimizes the crime rate is fair according to a metric that has attracted attention in the literature: setting \emph{different} thresholds on the posterior probability of crime for each group so as to guarantee equality of false positive and negative rates. This corresponds to setting the \emph{same} threshold on signals across groups.

To dig a little deeper into this result, the equilibrium crime rate in each population can be viewed as the adjudicator's prior belief in equilibrium that an agent has committed a crime, given knowledge only of her group membership. Given the noisy signal, the adjudicator has a posterior belief that the agent has committed a crime. In a static environment, the optimal classification rule for an adjudicator who wishes to optimize classification accuracy will be a group-independent threshold on her posterior belief that an agent has committed a crime. Note that priors will generally differ between populations (because outside option distributions differ, crime rates in the two populations differ). Therefore policies corresponding to group-independent thresholds on posterior beliefs will typically correspond to applying group-dependent thresholds on signals and vice versa.

Equalizing false-positive and false-negative rates across groups then corresponds to setting identical thresholds on the raw \emph{signal} for each group. It can be viewed as a commitment to avoid conditioning on group membership, even when group membership information is ex-post informative for classification. The intuition is that if the adjudicator uses the same threshold on the posterior belief that an agent has committed a crime for each group, the decision rule is making use of information contained within each group's prior. Although this information is statistically informative of the decision made by the agent, it is not within the agent's control. Using this information therefore only distorts the (dis-)incentives to commit crime. On the other hand, if the adjudicator uses the same threshold on agents' signals for each group (and hence different thresholds on posterior beliefs), decisions are made only as a function of information under the control of the agents, and hence are more effective at discouraging crime. The equalizing of false-positive and false-negative rates across the groups follows from this.

%\cl{One takeaway of our paper (discussed in the meeting): heterogeneous signal structures are crucial in discussing what kind of fairness is preferred.} \cj{Should this discussion of the heterogeneous signal structure go into the next subsubsection: extensions of the model?}
%\cj{From our discussion on 1/17: Possible criticism of our solution is that the more disadvantaged group will commit more crime, but if we set the equal risk threshold, more disadvantaged group will be incarcerated.}

\subsubsection*{Extensions}
The main insights of our baseline model continue to hold even when many of its core  assumptions are relaxed. We summarize them here.

First, the baseline model assumes that a signal is observed by the adjudicator for every agent (or equivalently, at equal rates across populations). There is significant empirical evidence, however, that this is often not the case: for example, arrest rates (and hence prosecution rates)  are substantially higher in minority populations for certain drug offenses, despite evidence that the underlying prevalence is more uniform across groups \citep{drugdisparities}. Section \ref{subsec:hetobs} introduces an elaboration of the baseline model based on \cite{Persico02}.

In this variant, the adjudicator must rely on an intermediary (police) to inspect agents and generate a signal. The adjudicator observes a signal from an individual only if the police inspect them, and individuals are punished only if they are inspected \emph{and} their signal crosses the adjudicator's threshold for establishing guilt. The police have their own objectives: to maximize the number of successful inspections (e.g. inspections that result in an arrest). Formally, we study the following game: The adjudicator commits to a decision rule on guilt/ innocence based on signals. Then police and agents play a simultaneous move game: The police choose an \emph{inspection intensity} for each group to maximize their objectives, given the adjudicator's rule and crime rates in each group, subject to an overall capacity constraint. Agents of different groups commit crime based on both the adjudicator's rule, and the police's inspection rate for their group. We show that similar results continue to hold in this model, i.e. the optimal rule for the adjudicator will continue to equalize the disincentive to commit crime across the groups. This will result in equalizing \emph{conditional} false-positive and false-negative rates across the groups, i.e. the rate conditional on being inspected.

In Section \ref{sec:hetsignals}, we consider a setting in which the signal distribution depends not just on the action chosen by the agent (crime or no crime) but may also depend on their group. For example, the underlying signal generating process may be less noisy for agents from certain groups, and noisier for others. Of course, in general, the structure of the optimal solution is closely tied to the relationship between the signal distributions, and our results cannot carry over without further assumptions.\footnote{Consider, for example, the case in which there is a perfectly informative signal for one group but not for another. Then, the optimal solution will have zero error for the former group, whereas error will be inevitable for the latter, for whom we only have a noisy signal.} Nevertheless, the insights provided by our previous analysis allow us to study the tradeoff between the adjudicator's objective (minimizing overall crime) and various fairness notions.  In particular, Theorem \ref{thm:het_eqInc} gives conditions under which our baseline insights  continue to hold, i.e. conditions under which the adjudicator's optimal rule will continue to equalize the disincentive to commit crime across groups.  Conversely,  Theorem \ref{thm:het_lem_comparison} shows conditions under which rules that equalize false-positive or false-negative rates across groups outperform rules that equalize incentives.

Finally, note that crime rate in our model (or what is referred to in the classification literature as ``base rate'') is endogenous, where it is normally modeled as exogenously fixed. This allows us to consider an additional notion of fairness that is ill-defined in the existing literature; namely classification rules that equalize base rate across groups. In Section \ref{sec:equalbase}, we study when this may be better or worse than the other notions we considered above in terms of the adjudicator's objective (overall crime rate).

\subsection{Additional Related Work}
The economic literature starting from \cite{arrow1972some} has  considered models of discrimination where agent decisions (e.g. to gain education) are endogenous and their incentives determined by a principal's choice (e.g. employer's hiring rule).  \cite{CL93,FV92} study models in which individuals have a choice about how much effort to exert, and identical populations can have different outcomes in (e.g. hiring markets) because of asymmetric self-confirming equilibria. There has also been extensive interest in the design of affirmative action policies for example in higher education. \cite{loury} makes the case that affirmative action may be necessary to correct historical inequity by constructing a dynamic model in which heterogeneity between two groups may persist if the principal uses a non-discriminatory rule going forward. The subsequent literature is too large to comprehensively cite, see \cite{fang2011theories} for a survey.

There has also been substantial interest in evaluating outcome data for evidence of discrimination, using (or developing) an underlying theoretical prediction of how such discrimination would manifest: see e.g. \cite{knowles2001racial}, \cite{Persico02} or \cite{anwar2006alternative} in the context of policing/ traffic stops. A large literature has studied lending data for evidence of discrimination against women and minorities: see e.g. \cite{ferguson1995constitutes} or \cite{ladd1998evidence} for overviews of both the debate on what measures of (un-)fairness to use, and an overview of the existing research.

More recently, in the computer science literature, several papers consider effort-based models that are similar in spirit to \cite{CL93,FV92}. \cite{HC18} propose a  two-stage  model of a labor market with a ``temporary'' (i.e. internship) and ``permanent'' stage, and study the equilibrium effects of imposing a  fairness constraint (``statistical parity'', which corresponds to hiring from two populations at equal rates) on the temporary stage. \cite{fat20} consider a model of the labor market with higher dimensional signals, and study equilibrium effects of ``subsidy'' interventions which can lessen the cost of exerting effort. \cite{downstream} study the effects of admissions policies on a two-stage model of education and employment, in which a downstream employer makes rational decisions, but student types are exogenously determined. Two recent papers \cite{delayed,delayed2} study non game-theoretic models by which classification interventions in an earlier stage can have effects on individual type distributions at later stages, and show that for many commonly studied fairness constraints (including several that we consider in this paper), their effects can either be positive or negative in the long term, depending on the functional form of the relationship between classification decisions and changes in the agent type distribution.

\renewcommand{\Re}{\mathbb{R}}

\section{Preliminaries} \label{sec:model}
\subsubsection*{Baseline Model}
Each agent belongs to a group $g \in \groups$. A group corresponds to some observable characteristic of the agent, for instance race or gender. There are $N_g$ agents in group $g$. For simplicity assume just two groups $\{1,2\}$ though the results extend straightforwardly to any finite number. Each agent makes a single binary decision to either commit a crime ($c$) or remain innocent ($i$).

Then, for each agent, the adjudicator observes a random signal $\sigVal \in \Re$ which is informative of the agent's guilt/innocence. The distribution of the signal depends only on whether the agent has committed a crime (and is therefore conditionally independent of their group). Criminals' signals are drawn according to the distribution $\crimeSigCDF$  (with pdf $\crimeSigPDF$) and innocents' signals drawn from $\noncrimeSigCDF$ (with pdf $\noncrimeSigPDF$). It is without loss of generality (reordering signals if necessary) to assume that the signal distributions satisfy the Monotone Likelihood Ratio Property (MLRP), i.e. higher signals imply a higher likelihood of guilt.

The adjudicator commits to a decision rule $\beta$, which labels an agent in group $g$ with signal $s$ as guilty with probability $ \probPolicy_g(s)\in [0,1] $. Note that implicitly this means the adjudicator perfectly observes an agent's group membership, along with the signal, at the time of adjudication. We write $\guilty=1$ to indicate that the agent is labeled guilty and $\guilty=0$ otherwise.

% okay, basically merged robustness section into this prelim section.
Now we describe agents' incentives to commit a crime in the first place. The agent receives a reward $ \crimeReward $ when he commits a crime, but pays a penalty of $ \crimeCost $ if he is labeled as guilty. An agent who does not commit a crime receives his outside option value $ \outOptVal $. All three quantities are privately known only to the agent and are drawn independently from a distribution that may potentially differ across the groups.  An agent in group $ g $ commits a crime if his net utility from committing a crime is higher than not:
\begin{align}
&\crimeReward - \crimeCost \Pr ( \guilty =1 \mid c, \group) \geq \outOptVal - \crimeCost \Pr(\guilty =1 \mid i, \group) \label{eqn:ic},
\intertext{which can be written as}
& \Pr ( \guilty =1 \mid c, \group) - \Pr ( \guilty =1 \mid i, \group)  \leq \frac{\crimeReward - \outOptVal}{\crimeCost}. \label{eqn:ic2}
\intertext{where $\frac{\rho-\omega}{\kappa} $ is the \emph{marginal benefit} of committing a crime normalized by the penalty. Define}
& \margCostCrime_g = \Pr(\guilty=1 \mid c, \group) - \Pr(\guilty=1 \mid i,g) \label{eqn:disincentive}
\end{align}
as the \emph{disincentive} for committing a crime: it is the group specific additional probability of being found guilty having committed a crime relative to not.  Then, the crime rate of group $g$ can be expressed in terms of $\outOptCdf_g$, the survivor function (i.e. 1-CDF) associated with the relevant quantity on the right hand side of \eqref{eqn:ic2} given the joint distribution of $\crimeReward, \crimeCost$ and $\outOptVal$:
\begin{align}
    	\crimeRate_g =  \Pr\left(  \margCostCrime_{g} \le \frac{\crimeReward - \outOptVal}{\crimeCost} \right) = \outOptCdf_g\left(  \margCostCrime_{g} \right).\label{eqn:crimerate}
\end{align}
 The adjudicator's objective  is to minimize the overall crime rate, i.e. to solve
\begin{align}
\min_{\beta \in B}	\sum_{g \in \groups} \numPeople_g  \crimeRate_g. \label{eqn:obj}\tag{OPT}
    \end{align}
where of course, $\beta$ determines $\Delta_g$ by \eqref{eqn:disincentive} and therefore $CR_g$ as per \eqref{eqn:crimerate}. Here, $B$ is the set of all feasible policies for the adjudicator, i.e. $B=\{\beta_g, g\in \groups: \beta_g: \Re \rightarrow [0,1]\}$.

\subsection{Fairness Measures}\label{fairmeasures}
We are interested in how decision rules $\beta$ respecting various notions of fairness perform relative to the optimal policy, and to each other. There are five main notions of fairness that we discuss throughout this paper, each of which corresponds to equalizing some statistical quantity across groups. Three of them have been considered both in the literature and in the popular press: equalizing false positive rates, false negative rates, and positive predictive value. These three notions of fairness are of particular interest to us because it has been shown that attaining all three measures simultaneously is impossible (\cite{KMR16,Chou17}).

Given a policy $ \beta_g $ for group $g$, true positive rate (TPR), false positive rate (FPR), false negative rate (FNR), and positive predictive value (PPV) are defined as
\begin{align}
    &\TPR_g = \Pr(\guilty=1 \mid \crime, g)=\int_{\mathbb{R}}\crimeSigPDF(\sigVal)\probPolicy_g(s)d\sigVal \label{eqn:tpr}\tag{TPR}\\
	&\FPR_g = \Pr(\guilty=1 \mid i,g ) =\int_{\mathbb{R}}\noncrimeSigPDF(\sigVal)\probPolicy_g(s)d\sigVal\label{eqn:FPR}\tag{FPR}\\
	&\FNR_g = \Pr(\guilty=0 \mid c,g) =\int_{\mathbb{R}}\crimeSigPDF(\sigVal)(1-\beta_g(\sigVal))d\sigVal \,\, (= 1- \TPR_g) \label{eqn:FNR}\tag{FNR}\\
	&\PPV_g = \Pr(\crime \mid \guilty=1, \group)=\frac{\crimeRate_g \TPR_g}{\crimeRate_g \TPR_g + (1- \crimeRate_g) \FPR_g} \label{eqn:PPV}\tag{PPV}
	\intertext{Note that in light of these definitions, we can rewrite \eqref{eqn:disincentive} as:}
	&\margCostCrime_g = \TPR_g - \FPR_g \label{eqn:disincentive2} \tag{$\margCostCrime$}
\end{align}

%\cl{Need to define/explain that equalizing $ \Delta_g $ is to equalize disincentive}
Additionally, we propose two {\em new} notions of fairness: equalizing disincentives (denoted $\margCostCrime$) and equalizing crime rates (denoted $\crimeRate$). We say the policy $\beta$ achieves fairness notion $\fairnessNotion \in \{\FPR, \FNR, \PPV, \margCostCrime, \crimeRate\}$ if the resulting respective quantity is the same across the groups when the adjudicator chooses policy $\beta$.  We write $ \Beta_{\fairnessNotion} $ to be the set of all policies that achieve fairness notion $\fairnessNotion$.

Given this framework we are interested in two questions. First, which of these fairness notions is compatible with the adjudicator's problem \eqref{eqn:obj}? Second, under what conditions is a particular fairness notion better than another in terms of the objective of minimizing overall crime rate, i.e. when do we have that for fairness notions $\fairnessNotion, \fairnessNotion'$:
\[
\min_{\beta \in \Beta_{\fairnessNotion}}
\sum_{g \in \groups} \numPeople_g  \crimeRate_g
\le
\min_{\beta \in \Beta_{\fairnessNotion'}}
\sum_{g \in \groups} \numPeople_g  \crimeRate_g.
\]

One final piece of notation will be useful. We denote by $\beta^\star$ the solution to the adjudicator's problem \eqref{eqn:obj}, i.e. the policy that minimizes crime overall. We sometimes refer to this as the optimal policy.  Further, for fairness notion $\fairnessNotion$, we denote as $\beta^\star_\fairnessNotion$ the solution to the adjudicator's problem among all rules that satisfy fairness notion $\fairnessNotion$,i.e. the rule that solves
\begin{align*}
    \min_{\beta \in \Beta_{\fairnessNotion}}
\sum_{g \in \groups} \numPeople_g  \crimeRate_g.
\end{align*}

\section{Results in the Baseline Model} \label{sec:baseline}
The main result we build around is that the solution to the adjudicator's problem \eqref{eqn:obj} is naturally ``fair'' in terms of three of the five measures above. It provides an interesting counterpoint to the impossibility results of \cite{KMR16} and \cite{Chou17}. Those results state that it is impossible to simultaneously equalize  false negative rates, false positive rates, and positive predictive value across groups. This raises a question of which of the fairness measures should be preferred over the others. By endogenizing the base rate of criminal activity, we find that equalizing false positive rates and equalizing false negative rates are preferred to equalizing positive predictive value in the sense that the former two are compatible with the optimal policy while the latter is not. Formally,

\begin{theorem}\label{thm:opt}
The adjudicator's optimal policy $\probPolicy^\star$ (i.e. the policy which solves \eqref{eqn:obj}) equalizes the disincentive to commit crime  \eqref{eqn:disincentive2} across groups. As a result, it also equalizes the false negative rates \eqref{eqn:FNR}, and false positive rates \eqref{eqn:FPR}.
\end{theorem}
%\cj{We have described what the crime-minimizing policy looks like in each setting, but we haven't shown why thresholding is optimal, have we?}

\begin{proof}
First note that because $\probPolicy_g$ can be set independently for each group $g$, minimizing the total crime rate is achieved by individually minimizing the crime rate within each group.  Recall that the crime rate within a group is $ \outOptCdf_g(\margCostCrime_g)$. This in turn is minimized by maximizing the disincentive of crime $ \margCostCrime_g $, since $\outOptCdf_g$ being a survivor function is non-increasing.  Recall that
\[
 \margCostCrime_g = 	\int_{\sigSet} (\crimeSigPDF(\sigVal)-\noncrimeSigPDF(\sigVal))\beta_g(s)d\sigVal.
\]
Therefore the optimal $\probPolicy_g$ is independent of the (group-dependent) distribution over private values defining $H_g$, and is therefore the same for all groups:
\[
	\probPolicy_g(s)=\begin{cases}
		1\text{ if } \crimeSigPDF(\sigVal)\geq \noncrimeSigPDF(\sigVal),\\
		0\text{ if } \crimeSigPDF(\sigVal) < \noncrimeSigPDF(\sigVal),
	\end{cases}
\]
Since the disincentive to commit crime is a function only of $\probPolicy_g$, this results in the same disincentive to commit crime at the optimal solution.

Finally note that both the $ \FNR_g $ and $\FPR_g $ for each group are a function only of $\noncrimeSigPDF(\sigVal)$ and $\crimeSigPDF(\sigVal)$ (which are identical across groups), and the chosen policies $\beta_g(\sigVal)$, which we have shown in the optimal solution will be identical across groups. Hence, the adjudicator's optimal policy will equalize false positive rates and  false negative rates across groups.
\end{proof}

Rather than thinking of an arbitrary function $\beta_g(\sigVal)$, it is more natural to think of the adjudicator as selecting a threshold $\sigThresh_g$ for each group $g$ so that any member of group $g$ whose signal $\sigVal$ exceeds $\sigThresh_g$ is labeled guilty. That is
\[
    \probPolicy_g(s)=\begin{cases}
		1\text{ if } s \ge \sigThresh_g\\
		0\text{ if } s < \sigThresh_g.
	\end{cases}
\]

\begin{remark}\label{thresholdopt}
Since, $ \crimeSigPDF$ and  $\noncrimeSigPDF $ satisfy the MLRP property  (i.e. $ \frac{\crimeSigPDF(\sigVal)}{\noncrimeSigPDF(\sigVal)} $ is non-decreasing in $ \sigVal $), $\beta^\star$ is a threshold policy by observation.  Note that if strict MLRP holds, then the optimal thresholds are unique.
\end{remark}

Under a threshold policy, group $g$'s true positive rate reduces to $\TPR_g =\crimeSigCDF(\sigThresh_g)$ and the false positive rate simplifies to  $\FPR_g =1-\noncrimeSigCDF(\sigThresh_g)$.

\subsection{Discussion}

%\cl{We can explore the problem of maximizing accuracy if desired. Aaron argued: `to arrest everyone always is the optimal policy if maximizing accuracy is the objective.' This is not a correct argument because, in our model, even if we arrest everyone no matter what, high outside option people will still choose not to commit a crime. So, it is possible that arresting everyone is not a solution.}

%\cl{Perhaps Aaron might be the best person to do this: Add discussions that contain: 1) equalizing incentives instead of posterior risk is optimal in our environment. 2) equalizing posterior risks is not individually fair. In contrast, equalizing incentives is individually fair in the sense that it treats / gives the same incentive the same.}

% \cl{The optimal policy equalizes incentives, false positive rates, and false negative rates. Therefore, 1. PPV is strictly worse (i think the above section mentions it) 2. this equalizes the thresholds, which necessarily (?) violates the Sharad's posterior risk thresholds.}

%\cj{In light of what Changhwa and I discussed, this part will need some serious revision.}
\subsubsection*{Posterior Thresholds}
It is interesting to contrast the policy of setting equal thresholds on the signal, which we show to be the optimal policy here, as opposed to equal thresholds on the `posterior' or another calibrated risk score. The latter is advocated in \cite{sharad2}, for example.

In that paper, the authors consider a setting where crime choices are exogenous/ fixed, and study the choice of policy that minimizes weighted misclassification  rates (i.e. acquittal of the guilty and incarceration of the innocent). They show that an optimal policy involves a common `threshold' on the posterior across groups, i.e.,  first, the adjudicator estimates the prior probability that an individual has committed a crime by considering the base rate of crime for the individual's group and then uses the observed signal to update her prior probability to her posterior belief that the individual in question has committed a crime. Second,  the individual is deemed guilty if the posterior probability of guilt exceeds some threshold.

Of course, in our setting, choice of crime is endogenous, and the planner's objective function is minimizing crime rate rather than minimizing  mislabeling costs. Nevertheless, it is interesting to inquire into the implications of equalizing the thresholds on the  posterior in our setting.

In our setting, the posterior after observing the signal $ s $ and the group $ g $ is
\[
    \Pr(c\mid s,g) = \frac{\crimeSigPDF(s)\crimeRate_g}{\crimeSigPDF(s) \crimeRate_g + \noncrimeSigPDF(s)(1-\crimeRate_g)}
\]
which increases in $ s $ when the signal structure satisfies monotone likelihood ratio property. Thresholding the posterior corresponds to choosing a value $ \pi_g\in [0,1] $ and classifying as guilty whenever $ \Pr(c\mid s,g)\geq \pi_g $.

Let $ T^* $ be the threshold on the signal under the planner's optimal policy. The optimal policy classifies the defendant as guilty if the signal $ s $ exceeds the threshold $ T^* $. With the monotone likelihood ratio property, this implies a threshold rule on the posterior which is to classify the defendant as guilty whenever the posterior exceeds
\[
    \pi_g = \Pr(c \mid T^*,g).
\]
By observation, the posterior thresholds are equalized across groups if and only if $\crimeRate_g = \crimeRate_{g'}$, i.e.  if and only if crime rates are equalized under the optimal policy $T^*$. This in turn will only occur if $H_g( \margCostCrime) = H_{g'} (\margCostCrime)$ which of course will not obtain in general because $H_g$ need not be the same as $ H_{g'}$.

\section{Heterogenous Signal Observation}\label{subsec:hetobs}

The baseline model assumes that when an agent commits a crime, the adjudicator observes the signal generated and adjudicates. What if the signals generated by members of each group are observed at different rates? This can happen if the adjudicator relies on an intermediary to record the signals. For example, in the crime application, the groups may be policed at different rates. Their (dis)incentives to commit crime will then differ as a result. Critically, we suppose that the police's incentives differ from the adjudicator's. The adjudicator's choice of rule therefore influences the police's choice on how to divide their manpower across different groups. Both the adjudicator's rule and the police's choice influence the incentives of agents to commit crime, and ultimately determine the overall crime levels in society.

To model this, we build upon upon the framework of \cite{Persico02}. There are a continuum of police officers who choose inspection intensities for each group $\{\inspecInsty_g\}_{g \in \groups}$ given a search capacity $\searchCap$. The choice of inspection intensity determines the rate at which signals are observed from the two groups: upon inspection, a police officer observes a signal about whether a crime was committed or not. In \cite{Persico02}, the signal is assumed to be perfect. We depart from this assumption in that in our setting, the observed signal is {\em noisy}, as in our baseline model. The adjudicator, as before, wishes to minimize the overall crime rate. As in \cite{Persico02}, the police have different incentives. Specifically, each police officer tries to maximize the number of `successful' inspections, i.e. where the signal recorded exceeds the threshold set by the adjudicator. As in \cite{Persico02}, we motivate this incentive as driven by the career concerns of individual police officers (who are e.g. promoted if they have many successful arrests etc.).

The timing of the game is therefore: (1) The adjudicator chooses the function $\probPolicy$, (2) the police takes this as fixed and chooses inspection intensities $ \inspecInsty_g\in [0,1] $ for each group subject to the constraint $ \numPeople_1 \inspecInsty_1 + \numPeople_2 \inspecInsty_2 \leq \searchCap $ (recall that $\numPeople_g$ is the number of agents in group $g$),\footnote{To make this model non trivial, we assume that search capacity is limited, i.e., $\searchCap < \numPeople_1 + \numPeople_2 $.} and then (3) given the adjudicator's choice $\probPolicy$, and inspection probability $ \inspecInsty $, an agent of group $g$ with crime reward $\crimeReward$, cost of being found guilty $\crimeCost$ and outside option commits a crime if (analogous to (\ref{eqn:ic}), but now taking into account also the probability of inspection):
\begin{align*}
    &\crimeReward - \inspecInsty_g \crimeCost  \Pr( q=1 \mid c, g) \geq \outOptVal - \kappa \inspecInsty_g \Pr (q=1 \mid i,g), \\
    \text{i.e., whenever }\,  & \inspecInsty_g \margCostCrime_g \leq \frac{\crimeReward - \outOptVal}{\crimeCost}
\end{align*}
where $\margCostCrime_g$ is as defined in \eqref{eqn:disincentive}.
By analogy with \eqref{eqn:crimerate} an agent of group $g$ commits crime with probability $H_g (\inspecInsty_g \margCostCrime_g)$.

As a benchmark, consider a setting where the adjudicator can choose both $\beta$ and $\theta$. Here the objective function is to minimize the overall crime rate just as in \eqref{eqn:obj}, and the additional constraint simply reflects that the choice of inspection rule must be feasible, i.e. the total level of inspection cannot exceed total search capacity $S.$ That is to say the adjudicator's problem in this benchmark can be written as:
\begin{equation}
		\begin{aligned}
			\min_{\{\beta, \theta\}} &\sum_g \numPeople_g \outOptCdf_g( \theta_g \margCostCrime_g) \\
			\text{s.t. } &\sum_g \numPeople_g \inspecInsty_g \leq \searchCap.
		\end{aligned}\label{persico-first-best}
	\end{equation}
We refer to the solution of (\ref{persico-first-best}) as the {\bf first-best} solution.

Now return to our setting above where the adjudicator chooses $\beta$ but not $\inspecInsty$. The police take $\beta$ as given and choose $\theta$ to maximize the number of successful inspections. Note that an inspection in group $g$ successful with probability $H_g(\inspecInsty_g \margCostCrime_g).$
As in \cite{Persico02}, we assume an interior equilibrium solution, i.e.,  $ \inspecInsty_g>0\ \forall g$.%
\footnote{A corner solution entails the police completely ignoring a group. In this case, the setting is trivial because conditional false positive rates are undefined for the ignored group.}
Intuitively, in this case, the optimal strategy for the police will equalize the crime rate between the two groups.  Otherwise, police that are trying to maximize their successful inspection probability will exclusively search the group with the highest crime rate.
%\ar{Why is the below objective the number of successful inspections? It looks like the overall crime rate. Shouldn't each term be multiplied by the respective inspection rate?} $\sum_g \numPeople_g \outOptCdf_g(\persicoOutOptThresh(\inspecInsty_g, \sigThresh_g)) $.  Formally, the police choose $\{\inspecInsty_g\}_g$ to solve
%\begin{align*}
%\max  \sum_{i} \numPeople_g \outOptCdf_g(\persicoOutOptThresh(\inspecInsty_g, \beta_g)) \\
%s.t.\,\, \sum_{i} \numPeople_g \inspecInsty_g = \searchCap
%\end{align*}
Recognizing this, the adjudicator will solve the following problem:
	\begin{equation}
		\begin{aligned}
			\min_{\beta, \theta} &\sum_g \numPeople_g \outOptCdf_g(\theta_g \margCostCrime_g)\\
			\text{s.t. } &\sum_g \numPeople_g \inspecInsty_g = \searchCap\\
			&\outOptCdf_1( \theta_1 \margCostCrime_1)=\outOptCdf_2( \theta_2 \margCostCrime_2).
		\end{aligned}\label{persico-second-best}
	\end{equation}		
The solution to problem \eqref{persico-second-best} is the {\bf second-best} solution.

We note that because groups will be inspected at different rates, the TPR and FPR  as we have defined them should correctly be called the \emph{conditional} true and false positive rates respectively, i.e., the rates conditional on being inspected.\footnote{We observe that these conditional rates are implicitly what has been studied in the fairness in machine learning literature, because these are the rates that can be computed from the data.}  Theorem \ref{thm:opt} now carries over mutatis mutandis, i.e. in both the first and second best solution, the thresholds will be set so as to equalize the conditional false and true positive rates $\CFPR$ and $\CTPR$. Proofs for this and subsequent theorems are in the appendix.

%The conditional disincentive of crime (CMC) is
%$$\Pr(g = 1 \mid c=1, h=1) - \Pr(g=1 \mid c=0, h=1) =(1-F(T-1)) - (1-F(T)) = F(T)-F(T-1) = \Delta_T .$$
%The unconditional disincentive of crime (UMC) is
%$$\sigma \cdot \Delta_T=\Pr(g=1 | c=1) - \Pr(g=1 | c=0).$$

%The unconditional false positive rate (UFPR) is defined to be
%$$\Pr(g=1 \mid c=0)=\sigma \cdot (1-F(T)).$$
%The while the unconditional false negative rate (UFNR) is
%$$\Pr(g=0 \mid c=1)=\sigma F(T-1) + (1-\sigma).$$
%In all cases the conditioning event is the event of being inspected. 	
\begin{restatable}{theorem}{thmPersicoEquivalenceFirstSecond}
\label{thm:persico-equivalence-first-second}
The optimal solutions to both the first best (\ref{persico-first-best}) and second best (\ref{persico-second-best}) equalize the $\CFPR$ and $\CTPR$ across groups.
%\ar{Weird that one of these optimization problems is denoted by a number and the other by a name}
\end{restatable}

While the optimal $\beta$ under the first and second best outcomes coincide, the optimal inspection intensities need not. In particular, the first- and second-best solutions coincide when $ \outOptCdf $ is convex. However, when $\outOptCdf$ is concave, then the inspection intensities under the second-best outcome \emph{maximize} the average number of crimes out of all possible search intensities given the optimal signal thresholds.

\begin{restatable}{theorem}{thminspection}
\label{inspection}
Suppose that the $ \outOptCdf_g $ belong to the same location family, i.e. $ \outOptCdf_g(\sigVal)=\outOptCdf(\sigVal-\mu_g) $ for some $ \mu_g $ for each $ i \in \groups$ and that $ \outOptCdf $ is convex (concave). Then, the inspection intensities in the second best solution minimize (maximize) the crime rate among all thresholds that equalize conditional false positive rates and conditional true positive rates $\CFPR$ and $\CTPR$.
\end{restatable}

%\cl{EDIT IT}

\section{Heterogeneous Signal Structure}
\label{sec:hetsignals}
We now examine the extent to which the conclusion of Theorem \ref{thm:opt} holds if we allow the signal structure $\sigStruct_g = (\crimeSigCDF_g,\noncrimeSigCDF_g)$ to be different across groups $g \in \groups$.
The signal structure $ \sigStruct_g $ and the  strategy $ \probPolicy_g(\sigVal) $ matter to the extent they discourage crime. Recall, that the relevant sufficient statistic of a strategy $\probPolicy_g(\sigVal)$ is what we called the \emph{disincentive} to commit crime: \[\margCostCrime_g = \TPR_g - \FPR_g = \int_{\sigSet} (\crimeSigPDF(\sigVal)-\noncrimeSigPDF(\sigVal))\probPolicy_g(\sigVal)d\sigVal.\]

The set of achievable disincentives given a signal structure is $ [\underline \Delta_g, \overline \Delta_g] $ where
\begin{align*}
    &\underline \Delta_g = \int_{S^-_g} (\crimeSigPDF_g(\sigVal)-\noncrimeSigPDF_g(\sigVal))ds,
    \;\text{where } S^-_g \equiv \{s: \crimeSigPDF_g(\sigVal)-\noncrimeSigPDF_g(\sigVal)<0\},\\
    & \overline \Delta_g = \int_{S^+}(\crimeSigPDF_g(\sigVal)-\noncrimeSigPDF_g(\sigVal))ds,
    \;\text{where } S^+_g \equiv  \{s: \crimeSigPDF_g(\sigVal)-\noncrimeSigPDF_g(\sigVal)>0\}.
\end{align*}
The relevant sufficient statistics for the signal structure $ (\crimeSigCDF_g,\noncrimeSigCDF_g) $ for group $g$ are its minimal and maximal disincentive $ \underline \Delta_g $ and $ \overline \Delta_g $ which determines the range of disincentives a classification rule is able to provide.

\newcommand{\sigVar}{\sigma}

\subsection{General Analysis}
In this section, we  give some insight into what happens when the signal structure varies across groups without making further assumptions on how the signal structure varies. First, as should be clear from the intuition previously, what really matters for our results in the baseline model is not that the signal distributions are identical across populations, but rather that the maximal disincentive $\overline \margCostCrime_g$ is the same across groups: if we have this, then the basic insight of Theorem \ref{thm:opt} holds as before (maximizing disincentives across groups). This is summarized in Theorem \ref{thm:het_eqInc}. Note that since the signal structures are different, the implication of Theorem \ref{thm:opt}, i.e. that FPR/ FNR will also be equalized across groups, will not hold in general.  Further, if the maximal disincentive differs, then in general the result does not hold, as we show in Example \ref{ex:het_eqInc}. Theorem \ref{thm:het_lem_comparison} then provides conditions under which various ``natural'' fair policies are ranked under the adjudicator's objective to minimize overall crime.

Let us start with analyzing the optimal policy that minimizes  average crime. As in Section \ref{sec:baseline},  average crime is minimized by maximizing the disincentive for crime in each group, which is attained by setting $ \probPolicy_g(\sigVal) $ such that $ \Delta_g = \overline \Delta_g $ for every $g$. Unlike in Section \ref{sec:baseline}, the optimal policy does not guarantee any of the fairness notions described in section \ref{fairmeasures} --- equalizing disincentives, equalizing false positive rates, equalizing false negative rates or equalizing positive predictive value --- when the signal structures differ across the groups.

An immediate observation is that when the signal structures have the same maximal disincentives $ \overline \Delta_g $, then the optimal effective policy equalizes disincentives.

\begin{restatable}{theorem}{thmhetEqInc}
\label{thm:het_eqInc}
    Suppose that $ \overline \Delta_1 = \overline \Delta_2 $. The adjudicator's optimal policy (i.e. the solution to \eqref{eqn:obj}) equalizes disincentives \eqref{eqn:disincentive2} across groups.
\end{restatable}

Theorem \ref{thm:het_eqInc} follows from the core insight of Theorem \ref{thm:opt} that the optimal rule maximizes the disincentive to commit crime for each group. When the signal structures across the groups are identical as in Theorem \ref{thm:opt}, the signal structures have the same maximal disincentives, and therefore, the optimal policy equalizes disincentives. Equalizing disincentives coincides with equalizing false positive rates and equalizing false negative rates in this case. However when the distributions of signals are different, equalizing disincentives need not be be the same as equalizing false positive/ false negative rates.  Indeed, equalizing disincentives may yield a strictly lower crime rate than equalizing false positive rates and equalizing false negative rates even when $ \overline \Delta_1=\overline \Delta_2 $.  To see this, consider the following example.

% \begin{example}
% \cl{Work on this figure later} Let us consider signal structures where $ \noncrimeSigPDF_1(s) = f(s) $ for some  distribution $ f $ that is asymmetric around $0$. Suppose $ \crimeSigPDF_1(s) = \noncrimeSigPDF_1(s-1) $ so that committing a crime produces a higher signal than not committing a crime by $ 1 $ in average. Finally, suppose $ f_2 $ is a horizontal reflection of $ f_1 $, that is, $ \noncrimeSigPDF_2(s) = f(-s) $ and $ \crimeSigPDF_2(s) = f(-s+1) $. Clearly, $ \overline \Delta_1 = \overline \Delta_2 $, and therefore, the optimal policy equalizes incentives. Let $ T_g^* $ be the threshold for group $g$ under the optimal policy. Since $ f_2 $ is a horizontal reflection of $ f_1 $, $ T_2^* = -(T_1^*-1) $ and $ T_2^*-1 = -T_1^* $ as in Figure \ref{}, which implies $ FPR_1 \neq FPR_2 $ and $ FNR_1\neq FNR_2 $. Consequently, equalizing false positive rates and equalizing false negative rates yield strictly higher crime rates than equalizing incentives. This also implies $ PPV_1 \neq PPV_2 $ generically so that equalizing PPVs also yields a strictly higher crime rate. \cl{genericity}
% \end{example}

\begin{example}\label{ex:het_eqInc}
    % \cl{Chris: Work on the figure as in my hand drawn figure (also, need to shade the area between $ T_g^*-1 $ and $ T_g^* $). The signal distribution for each group needs to be on the same x-axis (instead of separate pics), the areas need to be shaded with different colors. Then change the words 'COLORS' in the explanation according to the color that you've chosen.}
    Suppose that the signal for group $ g $, $ s_g $, is generated according to
    \[
        s_g = \eta_g + 1_{c}
    \]
    where $\eta_g $ is a random variable that has pdf $ f_g(\eta) $ that is strictly log-concave and has full support on $ \mathbb{R} $, and $ 1_{c} $ is an indicator function that equals $ 1 $ if and only if the agent has committed a crime. In words, committing a crime produces a signal that exceeds the signal from not committing a crime by at least $ 1 $ on average, while the underlying distribution $ f_g $ may differ across the groups. Note that
    \begin{align*}
        \noncrimeSigPDF_g(\sigVal) = f_g(\sigVal)\quad\text{and}\quad
        \crimeSigPDF_g(\sigVal) = f_g(\sigVal-1)
    \end{align*}
    and the strict log-concavity guarantees that the strict Monotone Likelihood Ratio Property between $ \noncrimeSigPDF $ and $ \crimeSigPDF $ is satisfied, that is, $ \frac{\crimeSigPDF_g(s)}{\noncrimeSigPDF_g(s)} $ strictly increases in $ s $, and that $ f_g $ is unimodal. Finally, suppose that $ f_1(\eta) $ is asymmetric around its mode, and that $ f_2(\eta) = f_1(-\eta)$ is a horizontal reflection of $ f_1 $.

    For each threshold $ T_g $, the corresponding disincentives satisfy $ \Delta_g(T_g) = F_g(T_g) - F_g(T_g-1) $. The threshold $ T_g^* $ maximizes $ \Delta_g(T_g) $ if and only if it equalizes the pdfs at $ T_g^* $ and $ T_g^*-1 $:
    \[
        f_g(T_g^*)=f_g(T_g^*-1).
    \]
    Graphically, $ T_g^* $ and $ T_g^*-1 $ are obtained as a pair of intersection points between the pdf $ f_g $ and a horizontal line, where the distance between the intersection points has to be $ 1 $ as in Figure \ref{fig:mirrored_sig_distr}. The maximal disincentive $ \overline\Delta_g = \Delta_g(T_g^*) $ is the white area under $ f_g $ between $ T_g^*-1 $ and $ T_g^* $. Since $ f_2 $ is merely a horizontal reflection of $ f_1 $, so is the maximal disincentives, and therefore, $ \overline \Delta_1 = \overline \Delta_2 $. By Theorem \ref{thm:het_eqInc}, the optimal policy $ T_g^* $ equalizes disincentives.

    The false positive rate and false negative rate for each group are colored in blue and red. Clearly, $ FPR_1 \neq FPR_2 $ and $ FNR_1\neq FNR_2 $. Consequently, equalizing false positive rates and equalizing false negative rates yield strictly higher crime rates than equalizing disincentives. This also implies $ PPV_1 \neq PPV_2 $ so that equalizing PPVs also yields a strictly higher crime rate in general.

    % With $ f_2 $ being a horizontal reflection of $ f_1 $, we claim that $ T_2^* = -(T_1^* - 1) $. To see this, note that at $ T_2 = -(T_1^* - 1) $,
    % \[
    %     f_2(T_2-1) = f_2(-T_1^*) = f_1(T_1^*) = f_1(T_1^*-1) = f_2(-(T_1^*-1)) = f_2(T_2)
    % \]
    % so that the pdfs are equalized at $ T_2 = -(T_1^*-1) $ and $ T_2-1 = -T_1^* $, implying that $ T_2^* = -(T_1^*-1) $. Graphically,

    % \begin{align*}
    %     \noncrimeSigPDF_1(\sigVal) &= f_1(\sigVal)\\
    %     \crimeSigPDF_1(\sigVal) &= f_1(\sigVal-1)\\
    %     \noncrimeSigPDF_1(\sigVal) &= f_2(\sigVal) = f_1(-\sigVal)\\
    %     \crimeSigPDF_2(\sigVal) &= f_2(\sigVal-1) = f_1(-(\sigVal-1))
    % \end{align*}
    % We can verify that $ T_2^* = -(T_1^*-1) $

    % Since $ f_2 $ is a horizontal reflection of $ f_1 $, $ T_2^* = -(T_1^*-1) $ and $ T_2^*-1 = -T_1^* $ as in Figure \ref{}, which implies $ FPR_1 \neq FPR_2 $ and $ FNR_1\neq FNR_2 $. Consequently, equalizing false positive rates and equalizing false negative rates yield strictly higher crime rates than equalizing incentives. This also implies $ PPV_1 \neq PPV_2 $ generically so that equalizing PPVs also yields a strictly higher crime rate.

\end{example}

\begin{figure}[h]
		\centering
		\resizebox{0.27\columnwidth}{!}{%
            \begin{tikzpicture}[font=\sffamily,
declare function={Gauss(\x,\y,\z,\u)=0.8 * 1/(\z*sqrt(2*pi))*exp(-((\x-\y+\u*(\x-\y)*sign(\x-\y))^2)/(2*\z^2));},
every pin edge/.style={latex-,line width=1.5pt},
every pin/.style={fill=yellow!50,rectangle,rounded corners=3pt,font=\small}]
\begin{axis}[
    every axis plot post/.append style={
    mark=none,
    samples=101},
    no markers,
    clip=false,
    axis y line=none,
    axis x line*=bottom,
    ymin=0,
    xtick={-31,-15,15,31},
	xticklabels={$\sigThresh^*_1-1$,$\sigThresh^*_2-1$, $\sigThresh^*_1$,$\sigThresh^*_2$},
	tick label style={font=\tiny},
	legend style={at={(0.5,-0.1)},anchor=north},
 	legend columns=2 	]
	\tikzset{
        hatch distance/.store in=\hatchdistance,
        hatch distance=10pt,
        hatch thickness/.store in=\hatchthickness,
        hatch thickness=2pt
    }

    \makeatletter
    \pgfdeclarepatternformonly[\hatchdistance,\hatchthickness]{flexible hatch}
    {\pgfqpoint{0pt}{0pt}}
    {\pgfqpoint{\hatchdistance}{\hatchdistance}}
    {\pgfpoint{\hatchdistance-1pt}{\hatchdistance-1pt}}%
    {
        \pgfsetcolor{\tikz@pattern@color}
        \pgfsetlinewidth{\hatchthickness}
        \pgfpathmoveto{\pgfqpoint{0pt}{0pt}}
        \pgfpathlineto{\pgfqpoint{\hatchdistance}{\hatchdistance}}
        \pgfusepath{stroke}
    }
    \makeatother
    \addplot[name path = A, draw=none,domain=-40:40,forget plot] {0};

          \addplot[name path = B2, line width=1.5pt,black,domain=-40:40] {Gauss(-x,20,10,.5)};
    \addlegendentry{$f_1$}    
    \addplot[dashed, name path = B1, line width=1.5pt,black,domain=-40:40] {Gauss(x,20,10,.5)};
    \addlegendentry{$f_2$}

          \addplot[red!50, opacity=0.5] fill between[of=A and B2, soft clip={domain=-40:-31}];
	                          \addlegendentry{$\FNR_1$}
	                          
	    \addplot[draw,pattern=flexible hatch, pattern color=red,hatch distance=3pt,hatch thickness=0.4pt,
        area legend] fill between[of=A and B1, soft clip={domain=-40:-15}] \closedcycle;
                          \addlegendentry{$\FNR_2$}
                          
               	   \addplot[blue!50, opacity=0.5,area legend] fill between[of=A and B2, soft clip={domain=15:40}] \closedcycle;
                  \addlegendentry{$\FPR_1$}

     \addplot[mark=none,
        domain=-10:10,
        samples=100,
        pattern=flexible hatch,
        hatch distance=3.1pt,
        hatch thickness=0.41pt,
        draw=blue,
        pattern color=blue,
        area legend] fill between[of=A and B1, soft clip={domain=31:40}] \closedcycle;
        \addlegendentry{$\FPR_2$}

      \addplot[black,domain=-40:40,forget plot] {0};

\end{axis}
\end{tikzpicture}
        }
          \caption{}
        \label{fig:mirrored_sig_distr}
\end{figure}

When the signal structures across the groups have different maximal disincentives, we identify conditions under which both equalizing false positive rates and equalizing false negative rates yields a strictly lower crime rate than equalizing disincentives. Without loss of generality, let us assume that $ \overline\Delta_2 > \overline \Delta_1 $.

\begin{restatable}{theorem}{thmHetLemComparison}
\label{thm:het_lem_comparison}
    Let that $ \overline \margCostCrime_2 > \overline \margCostCrime_1 $.  Then, the following are equivalent:
    \begin{enumerate}
        \item The optimal policy subject to equalizing false positive rates ($\probPolicy_{\FPR}^\star$) attains a (weakly) lower crime rate than equalizing disincentives ($\probPolicy_{\margCostCrime}^\star$).
        \item The optimal policy subject to equalizing false negative rates ($\probPolicy_{\FNR}^\star$) attains a (weakly) crime rate than equalizing disincentives ($\probPolicy_{\margCostCrime}^\star$).
        \item $ (\crimeSigCDF_{2})^{-1} \circ \crimeSigCDF_1 (\sigThresh_1^*) > (\geq) (\noncrimeSigCDF_2)^{-1} \circ \noncrimeSigCDF_1(\sigThresh_1^*) $  where $ \sigThresh_g^* $ is the threshold under the optimal policy for group $ g $.
    \end{enumerate}
\end{restatable}

%\cl{Need help on the interpretation of the condition (3):} The last condition is quantile-quantile function  (similar, but different condition was used in Persico.) (the last condition holds true for any location-scale family, implying that equalizing FPR and FNR is always better than equalizing disincentives for all $ (G_g)_g $)

In general, the optimal policy does not guarantee any of the fairness notions. Theorem \ref{thm:het_lem_comparison} provides a sufficient and necessary condition under which equalizing false positive rates and equalizing false negative rates attains lower crime rates than equalizing disincentives. However, condition 3 is hard to interpret. Further, it is unclear which of equalizing false positive rate and false negative rate would be better overall for the adjudicator. Without additional structure on the signal structures, it is hard to proceed further. To explore these issues, we restrict attention to signal distributions that are members of location-scale families of distributions.

\subsection{Location Scale Families}

\begin{definition}\label{def:location-scale}
We say that the signal structure is from a location-scale family if each group's signal is
a location-scale transformation of the same underlying random variable $ \eta $ that has absolutely continuous and log-concave density function $ f $ with full support on the real line. Specifically, the signal $ \sigVal_g $ for group $g$ is generated according to:
\begin{equation*}
    \sigVal_g = \mu_g + \sigma_g \eta + m_g 1_c
\end{equation*}
where $ \mu_g $ is a location shifter, $ \sigma_g$ is a scale shifter, $ m_g$ is the marginal effect of crime on the signal and $ 1_c $ is an indicator function that equals $ 1 $ if and only if the agent has committed a crime. Equivalently, the conditional pdfs of signal $ s $ for group $ g $ conditioning on being innocent and having committed a crime are
\begin{equation}
    \begin{aligned}
    &\noncrimeSigPDF(\sigVal) = f\left(\frac{\sigVal-\mu_g}{\sigma_g}\right) \quad \text{and}\quad\crimeSigPDF(\sigVal) = f\left(\frac{\sigVal-\mu_g-m_g}{\sigma_g}\right).
    \end{aligned}
    \label{hetSig}
\end{equation}
\end{definition}

% \begin{definition}\label{def:location-scale}
% We say that the signal structure is from a location-scale family if each group's signal is
% a location-scale transformation of the same underlying random variable that has absolutely continuous and log-concave density function $ f $ with full support on the real line. Specifically, the signal for group $g$ is distributed as:
% \begin{equation}
%     \begin{aligned}
%     &\noncrimeSigPDF(\sigma_g \sigVal + \mu_g) = f(s),\\
%     &\crimeSigPDF(\sigma_g \sigVal + \mu_g + m_g) = f(s).
%     \end{aligned}
%     \label{hetSig}
% \end{equation}
% That is to say the underlying distributions are different in this parametric way, where $ \mu_g $ is a location shifter, $ \sigma_g$ is a scale shifter and $ m_g$ is the marginal effect of crime on the signal.
% \end{definition}

Note that the underlying distribution $ f $ is identical across the groups as in contrast to Example \ref{ex:het_eqInc} where the underlying distribution $ f_g $ differed across the groups. Combined with the functional form \eqref{hetSig}, log-concavity of $ f $ is equivalent to the signal structure satisfying the monotone likelihood ratio property for each group $g$ which implies that the optimal $ \beta $ is a threshold strategy. The log-concavity of $ f $ also implies that $ f $ is unimodal, which guarantees the uniqueness of the threshold that attains the optimal policy. There are many natural location-scale families of distributions satisfying log-concavity including normal distributions, logistic distributions, and extreme value distributions.

A property that makes location-scale family particularly tractable is that the disincentives engendered by a threshold depend only on $ \frac{m_g}{\sigma_g} $, i.e. is the ratio between the scale shift $\sigma_g$ and the marginal effect of crime $ m_g $. For this class of distributions, we can say that it is \emph{always} preferable to equalize either false positive rates or false negative rates compared to equalizing disincentives. %This establishes the following Theorem which is another generalization of Theorem \ref{thm:opt} by providing conditions on the signal structures under which the optimal policy equalizes incentives, false positive rates and false negative rates altogether.

%   then the crime minimizing signal thresholds $\{T^*_g\}_{i \in \mathcal{G}}$
\begin{restatable}{theorem}{thmscale}
\label{scale}
Suppose the distributions across groups are from the location-scale family as defined in Definition
\ref{def:location-scale}. Then \begin{enumerate}
    \item If $\frac{m_{1}}{\sigma_{1}} = \frac{m_{2}}{\sigma_{2}}$, then the optimal policy equalizes disincentives, false positive rates and negative rates.
    \item Suppose $ \frac {m_1} {\sigma_1} \neq \frac {m_2} {\sigma_2}$, and assume $\frac {m_2} {\sigma_2}$ is larger without loss of generality. Then $ \overline \Delta_2 > \overline \Delta_1 $. Further, the optimal policy subject to equalizing false positive rates ($\probPolicy_{\FPR}^\star$) and equalizing false negative rates ($\probPolicy_{\FNR}^\star$)  attain strictly lower crime rates than equalizing disincentives ($\probPolicy_{\Delta}^\star$). Further, all three attain strictly higher crime rates than the optimal policy ($\probPolicy^*$).
\end{enumerate}
% That is, $\forall i,j \in \mathcal{G},$
%\[1-x_j(T^*_g) = 1-x_{j}(T^*_{j})  \quad\text{and}\quad 1-y_g(T^*_g) = 1-y_{j}(T^*_{j}).\]
\end{restatable}
The formal proof is in the appendix.
For some intuition, note that the maximum disincentive $ \overline \Delta_g $ is determined by and increasing in $ \frac{m_g}{\sigma_g} $. Intuitively, this is because the larger the normalized marginal effect of crime on the signal is, the better the adjudicator is able to distinguish between the criminal and the non-criminal based on the signal, and therefore the adjudicator can increase the disincentive to commit crime. If this term is equal across groups, then Theorem \ref{thm:het_eqInc} applies and part (1) follows as a corollary. Now, without loss of generality, suppose instead that $ \frac {m_1} {\sigma_1} < \frac {m_2} {\sigma_2} $. Then $ \overline \Delta_2 > \overline \Delta_1 $. It can also be verified that the condition $ (3) $ in Theorem  \ref{thm:het_lem_comparison} holds, so that equalizing false positive rates and equalizing false negative rates always yield a lower crime rate than equalizing disincentives. Furthermore, we can also verify that equalizing false positive rates and equalizing false negative rates never attains the crime rate under the optimal policy.

% This is a generalization of Theorem \ref{opt} and a special case of Theorem \ref{thm:het_eqInc}. Location scale family is one class of signal distributions under which the optimal policy equalizes incentives, false positive rates and negative rates altogether, and hence, equalizing PPV necessarily yields a higher crime rate.

A natural question to ask is whether one of  equalizing false positive rates or equalizing false negative rates will have a lower crime rate than the other. It is a hard question to answer in general. When the underlying distribution $ f $ is symmetric around $0$ ($0$ is without loss of generality since the $\mu_g$'s can always be shifted if $f$ is symmetric around some other number), however, the set of feasible disincentives $ (\Delta_1, \Delta_2) $ are identical under equalizing false positive rates and false negative rates, and therefore, the two notions of fairness yield the same crime rate.

\begin{restatable}{theorem}{thmHetLocscaleFprfnrthesame}
\label{het_locscale_fprfnrthesame}
    Suppose the signal structure is from the location-scale family as in Definition \ref{def:location-scale}, and $ f $ is symmetric around $ 0 $. Then, the optimal policy subject to equalizing false positive rates and that subject to equalizing false negative rates yield the same crime rate.

\end{restatable}

\section{Equalizing Crime Rates}\label{sec:equalbase}

In our model, the crime rates are endogenously determined by agents' decisions as a response to the policy implemented. This is unlike most of the fairness literature that assumes that the underlying rates are fixed. It motivates us to study another fairness measure that previous papers could not have asked for: equalizing crime rates.

To understand the implications of equalizing crime rates, let us assume that group $ 2 $ is `riskier' than group $ 1 $ without loss of generality. Specifically, assume that $\outOptCdf_2$ stochastically dominates $\outOptCdf_1$ -- that is to assume $\outOptCdf_2(\margCostCrime) \ge \outOptCdf_1(\margCostCrime)\ \forall \margCostCrime$.

% let us assume that the outside option distribution for group $ 1 $, $ \outOptCdf_1 $, \ar{Is this consistent with the meaning of $\outOptCdf$ elsewhere in the paper? In addition to being a survival function earlier, isn't it instead the distribution of a quantity that depends on the outside option, and the cost and benefit of crime?} first-order stochastically dominates that of group $ 2 $, $ \outOptCdf_2 $. That is, $ \outOptCdf_1(\outOptVal) \leq \outOptCdf_2(\outOptVal) \ \forall \outOptVal\in \mathbb{R} $, meaning group $ 2 $ is more `disadvantaged' than group $ 1 $ or group $ 1 $ is more `privileged' than group $ 2 $.

% If the distribution of outside options is not ranked by by FOSD, the group that enjoys the lowest crime rate is the group that has the largest value of $ \outOptCdf_g $ at $ \crimeReward-\crimeCost\Phi^* $. In Figure \ref{bb}, it is group $ A$ that is `riskiest' at $ \crimeReward-\crimeCost\Phi^* $ and therefore has the lower crime rate. \ar{I'm confused by the basic assertions in this section, starting from why we need to equalize the posterior probability of crime conditioned on every signal, and therefore everything following.} \ch{More explanations    singal threhodls are in the income thresholds which determines feasible crime rates.}

In this section, we focus on the comparison between equalizing crime rates and equalizing disincentives, while allowing any arbitrary signals structures $\sigStruct_g = (\crimeSigCDF_g, \noncrimeSigCDF_g)$. Equalizing disincentives is an appropriate fairness measure to compare because it is attained by the optimal policy when $ \overline \margCostCrime_1 = \overline \margCostCrime_2 $.

The first question to ask is whether equalizing crime rates can ever attain a lower crime rate than equalizing disincentives. We find that equalizing crime rates attains a lower crime rate than equalizing disincentives if and only if $ \overline \Delta_2 $ is sufficiently larger than $ \overline \Delta_1 $.

%The main results from this section can be summarized as follows: (i) equalizing base rates attains a lower crime rate than equalizing disincentives if and only if $ \bar \Delta_2 $ is sufficiently larger than $ \bar \Delta_1 $;  (ii) it is possible that the optimal policy equalize base rates, while it does not equalize disincentives.

%More specifically, we show that equalizing base rates attains a lower crime rate than equalizing disincentives if and only if the signal structure for the riskier group can provide sufficiently higher maximal disincentive than that for the more privileged group, that is, $ \overline \margCostCrime_2 $ is sufficiently larger than $ \overline \margCostCrime_1 $. the optimal policy equalizes base rates when $ \overline \margCostCrime_2 $ is appropriately higher than $ \overline \margCostCrime_1 $. They are formally stated in Theorem \ref{eqBaseRates_eqIncen} and Theorem \ref{eqBaseRates_opt}.

\begin{restatable}{theorem}{thmEqBaseRatesEqIncen}
\label{eqBaseRates_eqIncen}
    Suppose that $ \outOptCdf_2 $ first-order stochastically domiantes $ \outOptCdf_1 $. Then, there is an $ \epsilon>0 $ such that equalizing crime rates attains a lower crime rate than equalizing disincentives if and only if $ \overline\Delta_2\geq \overline\Delta_1+\epsilon $.
\end{restatable}

\begin{figure}
\centering
		\resizebox{0.4\columnwidth}{!}{%
			\begin{tikzpicture}
			\begin{axis}[%
			xlabel=$\margCostCrime$,
			ylabel=$\outOptCdf_g(\margCostCrime_g)$,
			legend style={at={(0.5,-0.1)},anchor=north},
 			legend columns=2
			%legend pos=south east
			]
			\addplot [smooth, red] {normcdf(x,2,2)};
			\addlegendentry{$\outOptCdf_1$}
			\addplot [smooth, blue] {normcdf(x,0,2)};
			\addlegendentry{$\outOptCdf_2$}

			\addplot [only marks, mark = triangle*, mark options={rotate=90},mark size=5pt, blue, forget plot] coordinates {(4, {normcdf(4,0,2)})};
			%\addplot [only marks, mark = triangle*, mark options={rotate=-90},mark size=5pt, red, forget plot] coordinates {(-4, {normcdf(-4,2,2)})};
			%\addplot [only marks, mark = triangle*, mark options={rotate=-90},mark size=5pt, blue, forget plot] coordinates {(-2.5, {normcdf(-2.5,0,2)})};
			\addplot [only marks, mark = triangle*, mark options={rotate=90},mark size=5pt, red, forget plot] coordinates {(0, {normcdf(0,2,2)})};
			
			\addplot [only marks,mark = ,mark options={fill=white}, mark size = 2pt] coordinates {(0, {normcdf(0,0,2)})};
			\addlegendentry{Equalizing Incentives}
			\addplot [only marks,mark = +,mark options={fill=black, rotate=45}, mark size = 8pt] coordinates {(2.0, {normcdf(2.0,0,2)})};
			\addlegendentry{Equalizing Crime Rates}
			
			\addplot [only marks,mark = ,mark options={fill=white}, mark size = 2pt] coordinates {(0, {normcdf(0,2,2)})};
			\addplot [only marks,mark = +,mark options={fill=black, rotate=45}, mark size = 8pt] coordinates {(0, {normcdf(0,2,2)})};

			\addplot [very thick, red, domain=-4:0] {normcdf(x,2,2)};
			\addplot [very thick, blue, domain=-2.5:4] {normcdf(x,0,2)};
			
			\addplot [thick, orange] {normcdf(0,2,2)};
			
			\addplot[thick, black] coordinates {(0,0)(0,1)};

			%{(-2,normcdf(-1.5,0.5,2))};

			\end{axis}
			\end{tikzpicture}
			}
			\caption{$ \overline \margCostCrime_1 + \epsilon < \overline \margCostCrime_2 $}
			\label{eqBase_1}
\end{figure}

\begin{figure}
\centering
\begin{subfigure}{0.4\columnwidth}
    %\centering
    		\resizebox{1.0\columnwidth}{!}{%
	\begin{tikzpicture}
	\begin{axis}[%
	xlabel=$\margCostCrime$,
	ylabel=$\outOptCdf_g(\margCostCrime_g)$,
	legend entries={$\outOptCdf_1$, $\outOptCdf_2$},
 	legend pos=south west
	%grid=major,
	%legend entries={$\outOptCdf_1$, $\outOptCdf_2$},
	%legend style={at={(0.5,-0.1)},anchor=north},
	%		legend columns=2
	]
	\addplot [smooth, red] {normcdf(x,2,2)};
			\addlegendentry{$\outOptCdf_1$}
			\addplot [smooth, blue] {normcdf(x,0,2)};
			\addlegendentry{$\outOptCdf_2$}

			\addplot [only marks, mark = triangle*, mark options={rotate=90},mark size=5pt, blue, forget plot] coordinates {(1.25, {normcdf(1.25,0,2)})};
			%\addplot [only marks, mark = triangle*, mark options={rotate=-90},mark size=5pt, red, forget plot] coordinates {(-4, {normcdf(-4,2,2)})};
			%\addplot [only marks, mark = triangle*, mark options={rotate=-90},mark size=5pt, blue, forget plot] coordinates {(-2.5, {normcdf(-2.5,0,2)})};
			\addplot [only marks, mark = triangle*, mark options={rotate=90},mark size=5pt, red, forget plot] coordinates {(0, {normcdf(0,2,2)})};
			
			\addplot [only marks,mark = ,mark options={fill=white}, mark size = 2pt] coordinates {(0, {normcdf(0,0,2)})};

			\addplot [only marks,mark = +,mark options={fill=black, rotate=45}, mark size = 8pt] coordinates {(1.25, {normcdf(1.25,0,2)})};

			\addplot [only marks,mark = ,mark options={fill=white}, mark size = 2pt] coordinates {(0, {normcdf(0,2,2)})};
			\addplot [only marks,mark = +,mark options={fill=black, rotate=45}, mark size = 8pt] coordinates {(-0.75, {normcdf(-0.75,2,2)})};

			\addplot [very thick, red, domain=-4:0] {normcdf(x,2,2)};
			\addplot [very thick, blue, domain=-2.5:1.25] {normcdf(x,0,2)};
			
			\addplot [thick, orange] {normcdf(1.25,0,2)};
			
			\addplot[thick, black] coordinates {(0,0)(0,1)};
	\end{axis}
	\end{tikzpicture}
	}
	\caption{$ \overline \margCostCrime_1 < \overline \margCostCrime_2 < \overline \Delta_1 + \epsilon $}
	\label{eqBase_2}
\end{subfigure}%
\begin{subfigure}{0.4\columnwidth}
  \centering
  		\resizebox{1.0\columnwidth}{!}{%
  \begin{tikzpicture}
			\begin{axis}[%
			xlabel=$\margCostCrime$,
			ylabel=$\outOptCdf_g(\margCostCrime_g)$,
			legend entries={$\outOptCdf_1$, $\outOptCdf_2$},
			legend pos=south east]
				\addplot [smooth, red] {normcdf(x,2,2)};
			\addlegendentry{$\outOptCdf_1$}
			\addplot [smooth, blue] {normcdf(x,0,2)};
			\addlegendentry{$\outOptCdf_2$}

			\addplot [only marks, mark = triangle*, mark options={rotate=90},mark size=5pt, blue, forget plot] coordinates {(0, {normcdf(0,0,2)})};
			%\addplot [only marks, mark = triangle*, mark options={rotate=-90},mark size=5pt, red, forget plot] coordinates {(-4, {normcdf(-4,2,2)})};
			%\addplot [only marks, mark = triangle*, mark options={rotate=-90},mark size=5pt, blue, forget plot] coordinates {(-2.5, {normcdf(-2.5,0,2)})};
			\addplot [only marks, mark = triangle*, mark options={rotate=90},mark size=5pt, red, forget plot] coordinates {(0, {normcdf(0,2,2)})};
			
			\addplot [only marks,mark = ,mark options={fill=white}, mark size = 2pt] coordinates {(0, {normcdf(0,0,2)})};

			\addplot [only marks,mark = +,mark options={fill=black, rotate=45}, mark size = 8pt] coordinates {(0, {normcdf(0,0,2)})};

			\addplot [only marks,mark = ,mark options={fill=white}, mark size = 2pt] coordinates {(0, {normcdf(0,2,2)})};
			\addplot [only marks,mark = +,mark options={fill=black, rotate=45}, mark size = 8pt] coordinates {(-2, {normcdf(-2,2,2)})};

			\addplot [very thick, red, domain=-4:0] {normcdf(x,2,2)};
			\addplot [very thick, blue, domain=-2.5:0] {normcdf(x,0,2)};
			
			\addplot [thick, orange] {normcdf(0,0,2)};
			
			\addplot[thick, black] coordinates {(0,0)(0,1)};			
			\end{axis}
			\end{tikzpicture}
			}
  \caption{$ \overline \margCostCrime_1 = \overline \margCostCrime_2 $}
  \label{eqBase_3}
\end{subfigure}%

\begin{subfigure}{0.4\columnwidth}
  \centering
  		\resizebox{1.0\columnwidth}{!}{%
  \begin{tikzpicture}
			\begin{axis}[%
			xlabel=$\margCostCrime$,
			ylabel=$\outOptCdf_g(\margCostCrime_g)$,
			legend entries={$\outOptCdf_1$, $\outOptCdf_2$},
			legend pos=south east]
				\addplot [smooth, red] {normcdf(x,2,2)};
			\addlegendentry{$\outOptCdf_1$}
			\addplot [smooth, blue] {normcdf(x,0,2)};
			\addlegendentry{$\outOptCdf_2$}

			\addplot [only marks, mark = triangle*, mark options={rotate=90},mark size=5pt, blue, forget plot] coordinates {(-0.6, {normcdf(-0.6,0,2)})};
			%\addplot [only marks, mark = triangle*, mark options={rotate=-90},mark size=5pt, red, forget plot] coordinates {(-4, {normcdf(-4,2,2)})};
			%\addplot [only marks, mark = triangle*, mark options={rotate=-90},mark size=5pt, blue, forget plot] coordinates {(-2.5, {normcdf(-2.5,0,2)})};
			\addplot [only marks, mark = triangle*, mark options={rotate=90},mark size=5pt, red, forget plot] coordinates {(0, {normcdf(0,2,2)})};
			
			\addplot [only marks,mark = ,mark options={fill=white}, mark size = 2pt] coordinates {(-0.6, {normcdf(-0.6,0,2)})};

			\addplot [only marks,mark = +,mark options={fill=black, rotate=45}, mark size = 8pt] coordinates {(-0.6, {normcdf(-0.6,0,2)})};

			\addplot [only marks,mark = ,mark options={fill=white}, mark size = 2pt] coordinates {(-0.6, {normcdf(-0.6,2,2)})};
			\addplot [only marks,mark = +,mark options={fill=black, rotate=45}, mark size = 8pt] coordinates {(-2.6, {normcdf(-2.6,2,2)})};

			\addplot [very thick, red, domain=-4:0] {normcdf(x,2,2)};
			\addplot [very thick, blue, domain=-2.5:-0.6] {normcdf(x,0,2)};
			
			\addplot [thick, orange] {normcdf(-0.6,0,2)};
			
			\addplot[thick, black] coordinates {(-0.6,0)(-0.6,1)};			
			\end{axis}
			\end{tikzpicture}
			}
  \caption{$ \overline \margCostCrime_1 > \overline \margCostCrime_2 $}
  \label{eqBase_4}
\end{subfigure}%
\caption{}
\label{fig:eql_base_rates_combined}
\end{figure}

The theorem is best demonstrated using the diagrams, although the formal proof is provided in the appendix. For theorem \ref{eqBaseRates_eqIncen}, we consider 4 cases: (i) $\overline{\margCostCrime}_1 + \epsilon \le \overline{\margCostCrime}_2$, (ii) $\overline{\margCostCrime}_1 < \overline{\margCostCrime}_2 <  \overline{\margCostCrime}_1 + \epsilon$, (iii) $\overline{\margCostCrime}_1 =\overline{\margCostCrime}_2$, and (iv) $\overline{\margCostCrime}_1 > \overline{\margCostCrime}_2$. For initial illustration purposes, we will focus on figure \ref{eqBase_1} which corresponds to the first case, $\overline{\margCostCrime}_1 + \epsilon < \overline{\margCostCrime}_2$. Each red and blue curve represents $ \outOptCdf_1(\cdot) $ and $ \outOptCdf_2(\cdot) $, respectively. Note the `thicker' segments of each curve are the set of crime rates $\outOptCdf_g(\margCostCrime_g)$ that can be achieved by varying the disincentive, $\margCostCrime_g \in [\underline \margCostCrime_g, \overline \margCostCrime_g] $. For each group, we denote the optimal policy by a triangle. The optimal policy while equalizing disincentives, which is denoted by `X', is obtained as intersections of the black line and the outside option distribution functions. The optimal policy while equalizing crime rates, which is denoted by `o', is obtained as intersections of the orange line and the outside option distribution function. In figure \ref{eqBase_1}, note that for group $1$, the crime rate stay the same under equalizing disincentives and crime rates, but the crime rate increases once one changes from the optimal policy that equalizes disincentives to the one that equalizes crime rates. Therefore, the optimal policy subject to equalizing crime rates is more preferred than under equalizing disincentives.

%Figure \ref{eqBase_1} depicts the case where $ \overline \margCostCrime_1 = \overline \margCostCrime_2 $. The optimal policies are depicted as the thick points on the left. This coincides with the optimal policy, while equalizing incentives, as alluded by Theorem \ref{thm:het_eqInc}. The optimal policy while equalizing base rates necessarily incurs an efficiency loss, as the the outside option threshold for group $1$ will need to move to the right, thereby increasing its base rate. \cl{Need to modify figures and add descriptions}.
%\cl{Need to modify figures and add descriptions}
%Figure \ref{eqBase_2} depicts the case when $ \overline \margCostCrime_1 > \overline \margCostCrime_2 $. Once again, equalizing incentives yields a lower crime rate than equalizing base rates, as the point for group $1$ will need to move to the right as the policy goes from the equalizing incentive to equalizing base rate.

Figure \ref{fig:eql_base_rates_combined} shows the other cases. Note that as we decrease $\overline{\margCostCrime}_2$ (or equivalently increase the crime rate of group $2$), there comes a point determined by $\numPeople_1$ and $\numPeople_2$ at which equalizing crime rates is more preferred than equalizing disincentives. More specifically, in figure \ref{eqBase_2}, imagine moving the right most blue triangle to the left and hence raising the orange line; as this happens, both `X' marks, which denote the optimal policy that equalizes crime rate, need to go up, while the optimal policy that equalizes incentive stays the same. Therefore, depending on the ratio of the number of people ($\numPeople_1$ and $\numPeople_2$), there exists some $\epsilon$ such that $\overline{\margCostCrime}_2 \le \overline{\margCostCrime}_1 + \epsilon$ if and only if equalizing crime rate attains lower crime rate than equalizing crime rates. And it's easy to see from figure \ref{eqBase_3} and \ref{eqBase_4} that equalizing disincentives achieves lower crime rate than equalizing crime rates in the corresponding cases. Therefore, these arguments together imply that equalizing crime rates attains a lower crime rate than equalizing disincentives if and only if $ \overline \margCostCrime_2 $ is sufficiently higher than $ \overline \margCostCrime_1 $.

% Figure \ref{eqBase_4} depicts the case when $ \overline \margCostCrime_2 $ is sufficiently larger $ \overline \margCostCrime_1 $ but smaller than $ \overline \margCostCrime_1 + \epsilon $. Then, the loss from each fairness notion is `this' and `that' (figures..), The loss from equalizing incentives decreases and the loss from equalizing base rates increases as $ \overline \margCostCrime_1 $ increases.

\begin{figure} [h]
			\centering
			\resizebox{0.4\columnwidth}{!}{%
			\begin{tikzpicture}
			\begin{axis}[%
			xlabel=$\margCostCrime$,
			ylabel=$\outOptCdf_g(\margCostCrime_g)$,
			legend entries={$\outOptCdf_1$, $\outOptCdf_2$},
			legend pos=south east]
			\addplot [smooth, red] {normcdf(x,2,2)};
			\addplot [smooth, blue] {normcdf(x,0,2)};
			\addplot [very thick, red, domain=-4:0] {normcdf(x,2,2)};
			\addplot [very thick, blue, domain=-2.5:2] {normcdf(x,0,2)};

			\addplot [thick, orange] {normcdf(0,2,2)};
			\addplot[thick, black] coordinates {(-0,0)(0,1)};
			
			%\addplot [only marks,mark = ,red] coordinates {(-4,{normcdf(-4,2,2)})};

			%{(-2,normcdf(-1.5,0.5,2))};
			%\addplot [only marks,mark = ,blue] coordinates {(-2.5, {normcdf(-2.5,0,2)})};
			\addplot [only marks, mark = triangle*, mark options={rotate=90},mark size=5pt, red] coordinates {(0, {normcdf(0,2,2)})};
			\addplot [only marks, mark = triangle*, mark options={rotate=90},mark size=5pt, blue] coordinates {(2, {normcdf(2,0,2)})};
			
			\end{axis}
			\end{tikzpicture}
			}
			\caption{$ \outOptCdf_1(\overline \margCostCrime_1) = \outOptCdf_2 (\overline \margCostCrime_2) $}
			\label{eqBase_5}
		\end{figure}

Having verified that equalizing crime rates can attain lower crime rates than equalizing disincentives, the next question is whether equalizing crime rates can ever attain a lower crime rate than any of other fairness notions. The answer to this question is positive which we establish by finding a condition under which the optimal policy attains equalizing crime rates.

\begin{restatable}{theorem}{eqBaseRatesOpt}
\label{eqBaseRates_opt}
    Suppose that $ \outOptCdf_2 $ first-order stochastically dominates $ \outOptCdf_1 $. When $ \outOptCdf_1( \overline{\margCostCrime}_1) = \outOptCdf_2(  \overline{\margCostCrime}_2) $, the optimal policy equalizes crime rates but not necessarily false positive rates, false negative rates or disincentives in general.
\end{restatable}

Figure \ref{eqBase_5} depicts the case when $ H_1(\overline \Delta_1) = H_2(\overline \Delta_2) $. By construction, the optimal policy equalizes crime rates. As it can be seen from Figure \ref{eqBase_5}, equalizing disincentives attains a strictly higher crime rates than equalizing crime rates. Furthermore, it can be shown that other notions of fairness --- equalizing false positive rates, false negative rates and positive predictive value - are not satisfied in general.

\section{Discussion and Conclusions} \label{sec:conclude}

This paper gives a general model in which classification rules which equalize false positive and false negative rates can be compatible with natural objectives, \emph{in spite of failing to capitalize on statistically relevant information}. We derived the model using the language of criminal justice, but one could just as easily apply the base model to settings in which the principle was making some other binary decision based on partial information, such as a lending or employment decision. The underlying reason is that conditioning on demographic information, while statistically useful, leads to decision rules that incentivize different groups differently --- \emph{because demographic information is not under individual control}. Hence, in settings in which the underlying objective depends on the decisions of rational agents, the decision rule should explicitly commit \emph{not} to condition on information that relates to an individual's demographic group, and instead use only information that is affected by the choices of the individual.  Abstracting away, the necessary conditions under which our conclusions hold are that:
\begin{enumerate}
    \item The underlying base rates are rationally responsive to the decision rule deployed by the principle,
    \item Signals are observed by the adjudicator at the same rates across populations, and
    \item The signals that the adjudicator must use to make her decision are conditionally independent of an individual's group, conditioned on the individual's decision.
\end{enumerate}
Here, conditions (2) and (3) are unlikely to hold precisely in most situations, but we give settings under which they can be relaxed.

More generally, if we are in a setting in which we believe that individual decisions are rationally made in response to the deployed classifier, and yet the deployed classifier does \emph{not} equalize false positive and negative rates, then this is an indication that \emph{either} the deployed classifier is sub-optimal (for the purpose of minimizing base rates), \emph{or} that one of conditions (2) and (3) fails to hold. Since in fairness relevant settings, the failure of conditions (2) and (3) is itself undesirable, this can be a diagnostic to highlight discriminatory conditions earlier in the pipeline than the adjudicator's decision rule. In particular, if conditions (2) or (3) fail to hold, then imposing technical fairness constraints on a deployed classifier may be premature, and instead attention should be focused on structural differences in the observations that are being fed into the deployed classifier.

%\newpage
%\bibliographystyle{econometrica}
\bibliographystyle{plainnat}
\bibliography{Refs}
%%%FOR NOW
%\end{document}

\newpage

\appendix
\numberwithin{lemma}{section}
\numberwithin{equation}{section}
\section{Omitted Results and Proofs}\label{sec:appendix}

\thmPersicoEquivalenceFirstSecond*
\begin{proof}[Proof of Theorem \ref{thm:persico-equivalence-first-second}]

We show that the optimal disincentives $\{\margCostCrime^*_g\}_g$ in both cases must be $\{\overline{\margCostCrime}_g\}_g$, which entails equalizing $\CFPR$ and $\CTPR$.

As for the first best outcome, it's easy to see that given any inspection intensities $\{\inspecInsty_g\}_{i \in \groups}$, minimizing the overall crime rate corresponds to maximizing $\{\margCostCrime_{i}\}_{i}$. Then, it follows that for the optimal set of search intensities and signal thresholds, the optimal solution should in the first best outcome should be such that $\margCostCrime^*_g = \overline{\margCostCrime}_g$ for each $g$.

Once again, for the second best outcome, given any feasible solution $(\{\inspecInsty_g\}_{g}, \{\margCostCrime_g\}_g)$ to \ref{persico-second-best}, we can show that if there exists $g$ such that $\margCostCrime_{g}$ is not maximized (i.e. $\margCostCrime_g < \overline{\margCostCrime}_g$), then we can always find a new feasible solution $(\{\inspecInsty'_g\}_g, \{\margCostCrime'_g\}_g)$ to \ref{persico-second-best} that sets $\margCostCrime'_{g} = \overline{\margCostCrime}_g$, while keeping other $\margCostCrime_g'$ the same  ($\margCostCrime'_{g'} = \margCostCrime_{g'}$) with strictly lower overall crime rate. This shows that the optimal solution to \ref{persico-second-best} must set $\margCostCrime^*_{g} = \overline{\margCostCrime}_g$ for each $g$.

Without loss of generality, assume that $\margCostCrime_{1} < \overline{\margCostCrime}_1$ is not maximized. Let's say $\margCostCrime'_{1} = \overline{\margCostCrime}_1 = (1+\epsilon) \margCostCrime_{1}$ for some $\epsilon > 0$ and $\margCostCrime'_2 = \margCostCrime_2$. Now, consider setting a new inspection intensity for group $1$ such that $\inspecInsty'_1 \in (\frac{1}{1+\epsilon}\inspecInsty_1, \inspecInsty_1)$ to guarantee that the crime rate in group $1$ will be strictly lower than before -- that is
\[\inspecInsty'_1 \margCostCrime'_1 > \inspecInsty_1 \margCostCrime_1 \quad\implies\quad \outOptCdf_1(\inspecInsty'_1 \margCostCrime'_1) < \outOptCdf_1(\inspecInsty_1 \margCostCrime_1).\] Then, because $\sum_g \numPeople_g \inspecInsty_g = \searchCap$, decreasing $\inspecInsty_1$ to $\inspecInsty'_1$ will require increasing $\inspecInsty_2$ to some $\inspecInsty'_2$. Then, the crime rate for group $2$ will necessarily decrease:
\[
\inspecInsty'_2 \margCostCrime'_2 > \inspecInsty_2 \margCostCrime_2 \quad\implies\quad \outOptCdf_2(\inspecInsty'_2 \margCostCrime'_2) < \outOptCdf_2(\inspecInsty_2 \margCostCrime_2).
\]
By the continuity of $\outOptCdf_g$, there exists $\inspecInsty'_1 \in (\frac{1}{1+\epsilon}\inspecInsty_A, \inspecInsty)$ such that $\outOptCdf_1(\inspecInsty'_1\beta'_1)) = \outOptCdf_2(\inspecInsty'_2\beta_2)$. Note that the crime rate in group 1 and group 2 must have decreased. Therefore, we have found a better feasible solution to \ref{persico-second-best}.

Also, note that this proof can be generalized even when the number of groups is greater than 2. We can aggregate a collection of groups whose $\margCostCrime_g$ are not changed to one `super' group, apply the same argument as above, and use induction over the number of groups.
\end{proof}

\thminspection*
\begin{proof}[Proof of Theorem \ref{inspection}]
We first provide a high level sketch of the proof. Using the fact that both $\outOptCdf_g$'s are from the same family (i.e. mean shifted), we show that the equilibrium inspection intensities (i.e. the second best solution) set the derivative of the objective value to 0. Now, using the the convexity (concavity) of $\outOptCdf$, we can show that the second derivative of the overall crime rate will be positive (negative), showing the equilibrium inspection intensities achieve local minima (maxima).

First, we show that if $\outOptCdf_g$'s belong to the same location family, then the derivative of the objective value evaluated at the equilibrium inspection intensities will be 0. Denote the equilibrium inspection intensities by $\{\inspecInsty_g^*\}_g$ and the equilibrium disincentives by $\{\margCostCrime^*_g\}_g$. Recall from Theorem \ref{thm:persico-equivalence-first-second} that $(\margCostCrime^*_1, \margCostCrime^*_2)$ in both the first and second best solution correspond to $(\overline{\margCostCrime}_1, \overline{\margCostCrime}_2)$.

The equilibrium inspection intensities equalize the crime rates. Hence, we have
\begin{align}
\label{eqn:inspec_equil_gnspec}
\outOptCdf_1(\inspecInsty^*_1 \margCostCrime_1) &= \outOptCdf_2(\inspecInsty^*_2 \margCostCrime_2) \cr
\Rightarrow \outOptCdf(\inspecInsty^*_1 \margCostCrime_1 - \mu_1) &= \outOptCdf( \inspecInsty^*_2 \margCostCrime_2 - \mu_2) \cr
\Rightarrow \inspecInsty^*_1 \margCostCrime_1 - \mu_1 &= \inspecInsty^*_2  \margCostCrime_2 - \mu_2 \cr
\Rightarrow \outOptPdf(\inspecInsty^*_1 \margCostCrime_1 - \mu_1) &= \outOptPdf( \inspecInsty^*_2 \margCostCrime_2 - \mu_2)\cr
\Rightarrow \outOptPdf_1(\inspecInsty^*_1 \margCostCrime_1) &= \outOptPdf_2( \inspecInsty^*_2 \margCostCrime_2)\cr
\end{align}

% \begin{align}
% \label{eqn:inspec_equil_gnspec}
% \outOptCdf_1(\crimeReward - \crimeCost \cdot \inspecInsty^*_1 \cdot \margCostCrime_1) &= \outOptCdf_2(\crimeReward - \crimeCost \cdot \inspecInsty^*_2 \cdot \margCostCrime_2) \cr
% \Rightarrow \outOptCdf(\crimeReward - \crimeCost \cdot \inspecInsty^*_1 \cdot \margCostCrime_1 - \mu_1) &= \outOptCdf(\crimeReward - \crimeCost \cdot \inspecInsty^*_2 \cdot \margCostCrime_2 - \mu_2) \cr
% \Rightarrow \crimeReward - \crimeCost \cdot \inspecInsty^*_1 \cdot \margCostCrime_1 - \mu_1 &= \crimeReward - \crimeCost \cdot \inspecInsty^*_2 \cdot \margCostCrime_2 - \mu_2 \cr
% \Rightarrow \outOptPdf(\crimeReward - \crimeCost \cdot \inspecInsty^*_1 \cdot \margCostCrime_1 - \mu_1) &= \outOptPdf(\crimeReward - \crimeCost \cdot \inspecInsty^*_2 \cdot \margCostCrime_2 - \mu_2)\cr
% \Rightarrow \outOptPdf_1(\crimeReward - \crimeCost \cdot \inspecInsty^*_1 \cdot \margCostCrime_1) &= \outOptPdf_2(\crimeReward - \crimeCost \cdot \inspecInsty^*_2 \cdot \margCostCrime_2).
% \end{align}

Replacing $\inspecInsty_2 = \frac{\searchCap - \numPeople_1 \inspecInsty_1}{\numPeople_2}$ and taking the derivative of the overall crime rate with respect to $\inspecInsty_1$ yields
\begin{align*}
&\numPeople_1 \margCostCrime_1 \outOptPdf_1(\inspecInsty_1\margCostCrime_1) + \numPeople_2 \outOptPdf_2\left(\left(\frac{\searchCap - \numPeople_1 \inspecInsty_1}{\numPeople_2}\right)\margCostCrime_2\right) \left(-\frac{\numPeople_1}{\numPeople_2}\margCostCrime_2\right) \cr
&= \numPeople_1 \margCostCrime_1\outOptPdf_1(\inspecInsty_1\margCostCrime_1) - \numPeople_1 \margCostCrime_2 \outOptPdf_2\left(\inspecInsty_2 \margCostCrime_2\right) \cr
\end{align*}
Note that by equation \ref{eqn:inspec_equil_gnspec} and $\margCostCrime^*_1 = \margCostCrime^*_2$, the derivative of the overall crime rate evaluates to 0 under $\{\inspecInsty_g^*\}_g$ and $\{\margCostCrime^*_g\}_g$.

Now, in order to determine whether $\{\inspecInsty_g^*\}_g$ achieves a local minima or maxima, we calculate the second derivative of the overall crime rate with respect to $\inspecInsty_1$:

\begin{align*}
    &\numPeople_1\margCostCrime_1^2 \outOptPdf'_1(\inspecInsty_1\margCostCrime_1) - \numPeople_1\margCostCrime_2 \outOptPdf'_2\left(\left(\frac{\searchCap - \numPeople_1\inspecInsty_1}{\numPeople_2}\right)\margCostCrime_2 \right) \left(-\frac{\numPeople_1}{\numPeople_2}\margCostCrime_2\right)\cr
    &=\numPeople_1\margCostCrime_1^2 \outOptPdf'_1(\inspecInsty_1\margCostCrime_1) + \frac{\numPeople_1^2}{\numPeople_2}\margCostCrime^2_2 \outOptPdf'_2\left(\left(\frac{\searchCap - \numPeople_1\inspecInsty_1}{\numPeople_2}\right)\margCostCrime_2 \right)
\end{align*}
By the convexity (concavity) of $\outOptPdf$, $\outOptPdf'_1$ and $\outOptPdf'_2$ is positive (negative). Therefore, $\{\inspecInsty_g^*\}_g$ achieves a local minima (maxima) at the equilibrium inspection intensities.
\end{proof}

\thmhetEqInc*
\begin{proof}[Proof of Theorem \ref{thm:het_eqInc}]
This follows directly from the fact that $(\overline{\margCostCrime}_1, \overline{\margCostCrime}_2)$ is the adjudicator's most optimal policy.
\end{proof}

% PROOF OF LEMMA \ref{thm:het_lem_comparison}
\thmHetLemComparison*
\begin{proof}[Proof of Theorem \ref{thm:het_lem_comparison}]

Let us begin with proving the following lemma.

\begin{lemma}\label{het_lemfprfnr}
Suppose that $ \overline \margCostCrime_2 > \overline \margCostCrime_1 $. Let $ \margCostCrime_g^{\FPR}, \margCostCrime_g^{\FNR}$ and $ \margCostCrime_g^{\margCostCrime} $ be the disincentives under the optimal policy subject to each fairness notion $\FPR$, $\FNR$, and $\margCostCrime$ respectively. Suppose further that the optimal policies while equalizing false positive rates, false negative rates and disincentives are threshold policies.
    \begin{enumerate}
        \item Equalizing false positive rates attains a (weakly) lower crime rate than equalizing disincentives for all $ (\outOptPdf_g)_g $ if and only if $ \margCostCrime_g^{\FPR} > (\geq) \margCostCrime_g^{\margCostCrime}\ \forall g $
        \item Equalizing false negative rates attains a (weakly) lower crime rate than equalizing disincentives for all $ (\outOptPdf_g)_g $ if and only if $ \margCostCrime_g^{\FNR} > (\geq) \margCostCrime_g^{\margCostCrime}\ \forall g $
    \end{enumerate}
\end{lemma}
\begin{proof}[Proof of lemma \ref{het_lemfprfnr}]
By definition, $ \margCostCrime_g^{\fairnessNotion} = \int_{\mathbb{R}}\crimeSigPDF(\sigVal)\probPolicy_g^\fairnessNotion(s)d\sigVal - \int_{\mathbb{R}}\noncrimeSigPDF(\sigVal)\probPolicy_g^\fairnessNotion(s)d\sigVal $ is the disincentive under the optimal policy $ \probPolicy_g^\fairnessNotion\in \arg\min_{\beta\in B_\fairnessNotion} \sum_{g \in \groups} \numPeople_g  \crimeRate_g $ that achieve each fairness notion $ \fairnessNotion\in \{\text{FPR,FNR},\Delta\} $

It is straightforward that the $ \margCostCrime_g^{\FPR} > (\geq) \margCostCrime_g^{\margCostCrime}\ \forall g $ implies equalizing false positive rates attains a (weakly) lower crime rate than equalizing disincentives. Note that $\outOptCdf_g$ is a non-increasing function, and therefore, for all $ g $,
\[
    \margCostCrime_g^{\FPR} > (\geq) \margCostCrime_g^{\margCostCrime} \iff \numPeople_g \outOptCdf_g(\margCostCrime_g^{\FPR}) < (\leq) \numPeople_g \outOptCdf_g(\margCostCrime_g^{\margCostCrime}).
\]
The same argument applies for $\FNR$.

Let us now prove the opposite direction: equalizing false positive rates attains a lower crime rate than equalizing disincentives for all $g$ implies that $\margCostCrime^{\FPR}_g > (\ge) \margCostCrime^{\margCostCrime}_g$ for all $g$. %\cl{ WORK ON THIS}\cl{WILL BE WORKING AGAIN FROM HERE}
To show this, suppose not. Without loss of generality, suppose that $ \margCostCrime^{\FPR}_1<(\leq)\margCostCrime^{\margCostCrime}_1 $. If $ \margCostCrime^{\FPR}_2\leq \margCostCrime^{\margCostCrime}_2 $, then any pair of survivor functions $ (H_1,H_2) $ that are strictly decreasing in their arguments will imply  $ H_1(\margCostCrime^{\FPR}_1) > H_1(\margCostCrime^{\margCostCrime}_1) $ and $ H_2(\margCostCrime^{\FPR}_2) \geq H_2(\margCostCrime^{\margCostCrime}_2) $ so that equalizing disincentives attains a strictly lower crime rate than equalizing false positive rates. If $ \margCostCrime^{\FPR}_2> \margCostCrime^{\margCostCrime}_2 $, then $ H_1 $ where that the difference between its value at $ \margCostCrime^{\FPR}_1 $ and at $ \margCostCrime^{\margCostCrime}_1 $ is large enough and $ H_2 $ where the difference between its value at $ \margCostCrime^{\FPR}_2 $ and at $ \margCostCrime^{\margCostCrime}_2 $ is small enough result in a lower crime rate for equalizing disincentives than for equalizing false positive rates. More specifically, let $ H_1 $ and $ H_2 $ be such that
\[
    N_1 (H_1(\Delta_1^{FPR}) - H_1(\Delta_1^{\margCostCrime})) > \epsilon
\]
and
\[
    N_2 (H_2(\Delta_2^{\margCostCrime}) - H_2(\Delta_2^{FPR})) < \epsilon
\]
for some $ \epsilon>0 $. Then,
\begin{align*}
    \Big(N_1 H_1(\Delta_1^{FPR}) + N_1 H_2(\Delta_1^{\margCostCrime})\Big) - \Big(N_1H_1(\Delta_1^{\margCostCrime}) + N_2 H_2(\Delta_2^{\margCostCrime}) \Big) > 0
\end{align*}
which states that equalizing disincentives attain a strictly lower crime rate than equalizing false positive rates. This completes the contradiction desired. Therefore, it has to be the case that $ \margCostCrime^{\FPR}_g<(\leq)\margCostCrime^{\margCostCrime}_g $ for all $ g $. The same argument may be applied for equalizing false negative rates.

% For the sake of contradiction, assume not. Then, then we can construct the income distributions $ H_g $ where equalizing incentives produces a strictly lower crime rate than equalizing false positive or negative rates does. For example, if we assume that $ \margCostCrime_g^{\FPR}\geq \margCostCrime_g^\margCostCrime $ but $ \margCostCrime_j^{\FPR}<\Delta_j^\margCostCrime $, then we can construct an income distribution where $ N_g (H_g(\Delta_g^{FPR}) - H_g(\Delta_g^{\margCostCrime})) < \epsilon $ and $ N_j (H_j(\Delta_j^{\margCostCrime}) - H_j(\Delta_j^{FPR})) > \epsilon $ for some $ \epsilon>0 $ so that the average crime rate is lower for $g$ than for $ j $.

\end{proof}
Because we are assuming signal threshold strategies by the adjudicator, we will parametrize $\FPR_g(\sigThresh_g), \FNR_g(\sigThresh_g), \TPR_g(\sigThresh_g)$ to denote false positive, false negative, and true positive rate when the signal threshold $\sigThresh_g$ is used.

We first show the equivalence between condition (1) and (3) in the theorem. Before we show the equivalence, we make some characterization of $\{\margCostCrime_g^\FPR\}_g$ and $\{\margCostCrime^\margCostCrime_g\}_g$. Since $ \overline \margCostCrime_2 > \overline \margCostCrime_1 $, it must be that \[ \margCostCrime_2^{\margCostCrime} = \margCostCrime_1^{\margCostCrime} = \overline \margCostCrime_1. \]
As for the first condition (1), by lemma \ref{het_lemfprfnr}, we have that $ \margCostCrime_g^{\FPR}\geq \margCostCrime_g^{\margCostCrime} $ for all $g$. Thus, we have that for (1),
\[
\margCostCrime_1^{\FPR} = \margCostCrime_1^{\margCostCrime} = \margCostCrime^\margCostCrime_1,
\]
which implies
\[
\sigThresh_1^{\FPR} =\sigThresh_1^{\margCostCrime} = \sigThresh_1^*.
\]

For $\FPR_1 = \FPR_2$, we have
\[
\FPR_2(\sigThresh_2^{\FPR}) = \FPR_1(\sigThresh_1^{\FPR}) = \FPR_1(\sigThresh_1^{\margCostCrime}) =  \noncrimeSigCDF_1(\sigThresh_1^*).
\]
or equivalently,
\[
\noncrimeSigCDF_2(\sigThresh_2^{\FPR}) = \noncrimeSigCDF_1(\sigThresh_1^{\FPR}) = \noncrimeSigCDF_1(\sigThresh_1^{\margCostCrime}) =  \noncrimeSigCDF_1(\sigThresh_1^*).
\]

Now, we show the equivalence:
\begin{align*}
    &\margCostCrime_2^{\FPR}\geq \margCostCrime_2^\margCostCrime\\
    &\iff \TPR_2(\sigThresh_2^{\FPR}) - \FPR_2(\sigThresh_2^{\FPR}) \ge \TPR_1(\sigThresh_1^{\margCostCrime}) - \FPR_1(\sigThresh_1^{\margCostCrime}) && \margCostCrime^\margCostCrime_2 = \margCostCrime_1^\margCostCrime\\
    &\iff \TPR_2(\sigThresh_2^{\FPR}) \ge \TPR_1(\sigThresh_1^{\margCostCrime}) &&\FPR_2(\sigThresh_2^{\FPR}) =  \FPR_1(\sigThresh_1^{\margCostCrime}) =  \FNR_1(\sigThresh_1^*)\\
    &\iff \crimeSigCDF_1(\sigThresh_1^{\margCostCrime}) \ge \crimeSigCDF_2(\sigThresh_2^{\FPR}) \\
    &\iff \crimeSigCDF_1(\sigThresh_1^*) \geq \crimeSigCDF_2(\sigThresh_2^{\FPR}) \\
    &\iff \crimeSigCDF_1(\sigThresh_1^*) \geq \crimeSigCDF_2((\noncrimeSigCDF_2)^{-1}\circ \noncrimeSigCDF_1(\sigThresh_1^{*})) && \sigThresh_2^{\FPR} = (\noncrimeSigCDF_2)^{-1}\circ \noncrimeSigCDF_1(\sigThresh_1^{\FPR}) = (\noncrimeSigCDF_2)^{-1}\circ \noncrimeSigCDF_1(\sigThresh_1^{*})\\
    &\iff (\crimeSigCDF_2)^{-1}\circ \crimeSigCDF_1(\sigThresh_1^*) \geq (\noncrimeSigCDF_2)^{-1}\circ \noncrimeSigCDF_1(\sigThresh_1^{*}) \\
\end{align*}

% \cj{FROM BEFORE}

% Since $ \overline \Delta_2 > \overline \Delta_1 $, it must be that $ \Delta_1^I = \Delta_2^I = \overline \Delta_1 $. For $ \Delta_g^{FPR}\geq \Delta_g^I $ for all $g$, $ \Delta_1^{FPR} = \Delta_1^I $ and $ T_1^{FPR} = T_1^I = T_1^* $. For $ FPR_2 = FPR_1 $, $ F^{nc}_2(T_2^{FPR}) = F^{nc}_1(T_1^{FPR}) = F^{nc}_1(T_1^*) $. Then,
% \begin{align*}
%     &\Delta_2^{FPR}\geq \Delta_2^I &\\
%     &\iff F_2^{nc}(T_2^{FPR}) - F_2^{cc}(T_2^{FPR}) \geq F_2^{nc}(T_2^I) - F_2^{cc}(T_2^I)&\\
%     &\iff F_1^{nc}(T_1^*) - F_2^{cc}(T_2^{FPR}) \geq F_1^{nc}(T_1^*) - F_1^{cc}(T_1^*)&(F_2^{nc}(T_2^{FPR})=F_1^{nc}(T_1^{FPR})=F_1^{nc}(T_1^*)\text{ and }\Delta_2^I=\Delta_1^I=\Delta_1^*)\\
%     &\iff F_1^{cc}(T_1^*) \geq F_2^{cc}(T_2^{FPR}) &\\
%     &\iff F_1^{cc}(T_1^*) \geq F_2^{cc}((F_2^{nc})^{-1}\circ F_1^{nc}(T_1^{*})) & T_2^{FPR} = (F_2^{nc})^{-1}\circ F_1^{nc}(T_1^{FPR}) = (F_2^{nc})^{-1}\circ F_1^{nc}(T_1^{*})\\
%     &\iff (F_2^{cc})^{-1}\circ F_1^{cc}(T_1^*) \geq (F_2^{nc})^{-1}\circ F_1^{nc}(T_1^{*}) &
% \end{align*}

A similar logic applies to equalizing false negative rates.
Once again, by lemma \ref{het_lemfprfnr}, $ \margCostCrime_g^{\FNR} \geq \margCostCrime_g^{\margCostCrime} $ for all $g$, which implies  $ \margCostCrime_1^{\FNR} = \margCostCrime_1^{\margCostCrime} $ and hence, \[ \sigThresh_1^{\FNR} = \sigThresh_1^{\margCostCrime} = \sigThresh_1^*. \]
For $ \FNR_2 = \FNR_1 $, we have
\[
    \TPR_2(\sigThresh^{\FNR}_2) = \TPR_1(\sigThresh^{\FNR}_1) = \TPR_1(\sigThresh^\margCostCrime_1) = \TPR_1(\sigThresh^*_1)
\]
or equivalently,
\[
    \crimeSigCDF_2(\sigThresh^{\FNR}_2) = \crimeSigCDF_1(\sigThresh^{\FNR}_1) = \crimeSigCDF_1(\sigThresh^\margCostCrime_1) = \crimeSigCDF_1(\sigThresh^*_1).
\]

The equivalence follows, as
\begin{align*}
    &\margCostCrime_2^{\FNR}\geq \margCostCrime_2^\margCostCrime\\
    &\iff \TPR_2(\sigThresh_2^{\FNR}) - \FPR_2(\sigThresh_2^{\FNR}) \ge \TPR_1(\sigThresh_1^{\margCostCrime}) - \FPR_1(\sigThresh_1^{\margCostCrime}) && (\margCostCrime^\margCostCrime_2 = \margCostCrime_1^\margCostCrime)\\
    &\iff - \FPR_2(\sigThresh_2^{\FNR}) \ge  - \FPR_1(\sigThresh_1^{\margCostCrime}) &&(\TPR_2(\sigThresh_2^{\FNR}) =  \TPR_1(\sigThresh_1^{\margCostCrime}) =  \TPR_1(\sigThresh_1^*))\\
    &\iff \noncrimeSigCDF_2(\sigThresh_2^{\FNR}) \ge \noncrimeSigCDF_1(\sigThresh_1^{\margCostCrime})\\
    &\iff \noncrimeSigCDF_2(\sigThresh_2^{\FNR}) \ge \noncrimeSigCDF_1(\sigThresh_1^{*})\\
    &\iff \noncrimeSigCDF_2((\noncrimeSigCDF_2)^{-1}\circ \crimeSigCDF_1(\sigThresh_1^{*})) \geq \noncrimeSigCDF_1(\sigThresh_1^{*})
    && (\sigThresh_2^{\FNR} = (\crimeSigCDF_2)^{-1}\circ \crimeSigCDF_1(\sigThresh_1^{\FNR}) = (\crimeSigCDF_2)^{-1}\circ \crimeSigCDF_1(\sigThresh_1^{*}))\\
    &\iff (\crimeSigCDF_2)^{-1}\circ \crimeSigCDF_1(\sigThresh_1^*) \geq (\noncrimeSigCDF_2)^{-1}\circ \noncrimeSigCDF_1(\sigThresh_1^{*})
\end{align*}
\end{proof}

\thmscale*
\begin{proof}[Proof of Theorem \ref{scale}] %Let $ \mathcal{S}=R $ be the support of $ \eta $.
\item\paragraph{Part (1)} Let $ \sigThresh^* $ be s.t.
	\[
		f(\sigThresh^*-r)=f(\sigThresh^*).
	\]
	Define $ \sigThresh_g = \mu_g + \sigma_g T^* $. Then,
	\begin{align*}
		-f_g(\sigThresh_g-m)+f_g(\sigThresh_g) &= -\frac 1\sigma_g f\Big(\frac{\mu_g+\sigma_g \sigThresh^* - \mu_g - m}{\sigma_g}\Big)+\frac 1\sigma_g f\Big(\frac{\mu_g+\sigma_g \sigThresh^* - \mu_g}{\sigma_g}\Big)\\
		&= -\frac 1\sigma_g f\Big(\sigThresh^*-\frac{m_g}{\sigma_g}\Big)+\frac 1\sigma_g f\Big(\sigThresh^*\Big)\\
		&= -\frac 1\sigma_g \Big(f(\sigThresh^*-r)+ f(\sigThresh^*)\Big)\\
		&=0.
	\end{align*}
	Therefore, $ \sigThresh_g^*=\sigThresh_g=\mu_g+\sigma_g \sigThresh^* $. Note that
	\[
		F_g(\sigThresh_g) = F(\sigThresh^*)
	\]
	and
	\[
		F_g(\sigThresh_g-m_g) = F\Big(\sigThresh^*-r\Big)
	\]
	for both $g$. That is, FNR and FPR are the same across the groups.
	
\paragraph{Part (2)} Suppose $ \frac {m_1} {\sigma_1} < \frac {m_2} {\sigma_2} $. Then $ \overline \Delta_2 > \overline \Delta_1 $.
    \[
        (F^{cc}_2)^{-1}\circ F^{cc}_1(T_1) = \Big(\frac{T_1-\mu_1-m_1}{\sigma_1}\Big)\sigma_2+\mu_2+m_2
    \]
    and
    \[
        (F^{nc}_2)^{-1}\circ F^{nc}_1(T_1) = \Big(\frac{T_1-\mu_1}{\sigma_1}\Big)\sigma_2+\mu_2
    \]
    so that
    \begin{align*}
        &(F^{cc}_2)^{-1}\circ F^{cc}_1(T_1) \geq (F^{nc}_2)^{-1}\circ F^{nc}_1(T_1)\\
        \iff & \Big(\frac{T_1-\mu_1-m_1}{\sigma_1}\Big)\sigma_2+\mu_2+m_2 \geq \Big(\frac{T_1-\mu_1}{\sigma_1}\Big)\sigma_2+\mu_2\\
        \iff & \frac{m_2}{\sigma_2} \geq \frac{m_1}{\sigma_1}.
    \end{align*}
    Therefore, by Theorem \ref{thm:het_lem_comparison}, equalizing false positive rates and equalizing false negative rates attains strictly lower crime rates than equalizing disincentives.
\end{proof}

\thmHetLocscaleFprfnrthesame*
\begin{proof}[Proof of Theorem \ref{het_locscale_fprfnrthesame}]
    Let $ (T_1,T_2) $ be thresholds that equalize false positive rates, that is, $ 1-F\Big(\frac{T_1-\mu_1}{\sigma_1}\Big)=1-F\Big(\frac{T_2-\mu_2}{\sigma_2}\Big) $. The disincentive for group $g$ is
    \begin{equation}
        F\Big(\frac{T_g-\mu_g}{\sigma_1}\Big) - F\Big(\frac{T_g-\mu_g-m_g}{\sigma_g}\Big).
        \label{deltaeq}
    \end{equation}

    Let $ T_g' = 2\mu_g + m_g - T_g $ for each $g$. Then,
    % \begin{align*}
    %     \frac{T_g'-\mu_g}{\sigma_g} &= -\frac{T_g-\mu_g-m_g}{\sigma_g}\\
    %     \frac{T_g'-\mu_g-m_g}{\sigma_g} &= -\frac{T_g-\mu_g}{\sigma_g}.
    % \end{align*}

    \[
        \frac{T_g'-\mu_g}{\sigma_g} = -\frac{T_g-\mu_g-m_g}{\sigma_g}
    \]
    and
    \[
        \frac{T_g'-\mu_g-m_g}{\sigma_g} = -\frac{T_g-\mu_g}{\sigma_g}.
    \]
    Note that
    \[
        F\Big(\frac{T_1'-\mu_1-m_1}{\sigma_1}\Big) = F\Big(-\frac{T_1-\mu_1}{\sigma_1}\Big) = 1-F\Big(\frac{T_1-\mu_1}{\sigma_1}\Big) = 1-F\Big(\frac{T_2-\mu_2}{\sigma_2}\Big) = F\Big(-\frac{T_2-\mu_2}{\sigma_2}\Big) = F\Big(\frac{T_2'-\mu_2-m_2}{\sigma_2}\Big)
    \]
    where the second and the fourth equalities are from the symmetry around $ 0 $, and the third equality is from $ (T_1,T_2) $ equalizing false positive rates. Therefore, $ (T_1',T_2') $ equalize false negative rates.

    Furthermore, the disincentive under $ T_g' $ is
    \begin{align*}
        F\Big(\frac{T_g'-\mu_g}{\sigma_1}\Big) - F\Big(\frac{T_g'-\mu_g-m_g}{\sigma_g}\Big) &= F\Big(-\frac{T_g-\mu_g-m_g}{\sigma_g}\Big) - F\Big(-\frac{T_g-\mu_g}{\sigma_g}\Big)\\
        &= \Big(1-F\Big(\frac{T_g-\mu_g-m_g}{\sigma_g}\Big)\Big) - \Big(1-F\Big(\frac{T_g-\mu_g}{\sigma_g}\Big)\Big)\\
        &= F\Big(\frac{T_g-\mu_g}{\sigma_g}\Big) - F\Big(\frac{T_g-\mu_g-m_g}{\sigma_g}\Big)
    \end{align*}
    which exactly is the disincentive under $ T_g $. Therefore, for any pair of disincentives that is feasible under equalizing false positive rates, it is feasible under equalizing false negative rates.

    Similar arguments can be applied to the other case. Therefore, the set of feasible pair of disincentives are identical, and therefore, the lowest crime rate that can be attained by equalizing false positive rates and equalizing false negative rates are identical.

    \end{proof}

\thmEqBaseRatesEqIncen*
\begin{proof}
We will fix $\outOptCdf_1$, $\outOptCdf_2$, and $\overline{\margCostCrime}_1$. Then, we will consider varying $\overline{\margCostCrime}_2$. There are 4 different cases: (i) $\overline{\margCostCrime}_1 + \epsilon \le \overline{\margCostCrime}_2$, (ii) $\overline{\margCostCrime}_1 < \overline{\margCostCrime}_2 <  \overline{\margCostCrime}_1 + \epsilon$, (iii) $\overline{\margCostCrime}_1 =\overline{\margCostCrime}_2$, and (iv) $\overline{\margCostCrime}_1 > \overline{\margCostCrime}_2$, where $\epsilon > 0$ is such that $\numPeople_1\left(\outOptCdf_2(\overline{\margCostCrime}_1 + \epsilon) - \outOptCdf_1(\overline{\margCostCrime}_1)\right) + \numPeople_2 \left(\outOptCdf_2(\overline{\margCostCrime}_1 + \epsilon) - \outOptCdf_2(\overline{\margCostCrime}_1)\right)  = 0$.
Such $\epsilon$ exists by the continuity of $\outOptCdf_g$. When $\epsilon=0$, then the above value is positive and once $\epsilon$ is big enough such that $\outOptCdf_2(\overline{\margCostCrime}_1 + \epsilon') = \outOptCdf_1(\overline{\margCostCrime}_1)$, then the value is negative. Therefore, by the intermediate value theorem, such $\epsilon$ exists. Furthermore, note that for $\epsilon' > 0$ whenever $\outOptCdf_2(\overline{\margCostCrime}_1 + \epsilon') < \outOptCdf_1(\overline{\margCostCrime}_1)$, it must be that $\numPeople_2 \left(\outOptCdf_2(\overline{\margCostCrime}_1 + \epsilon) - \outOptCdf_2(\overline{\margCostCrime}_1)\right) + \numPeople_1\left(\outOptCdf_2(\overline{\margCostCrime}_1 + \epsilon) - \outOptCdf_1(\overline{\margCostCrime}_1)\right) < 0$. Therefore, we have that $\outOptCdf_2(\overline{\margCostCrime}_1 + \epsilon) \ge \outOptCdf_1(\overline{\margCostCrime}_1)$.

%and $\outOptCdf_2(\overline{\margCostCrime}_1 + \epsilon) \ge \outOptCdf_1(\overline{\margCostCrime}_1)$.

Also, we write $\margCostCrime_\group^{\crimeRate}$ and $\margCostCrime_\group^{\margCostCrime}$ to denote the optimal disincentives that minimize the crime rates while equalizing the crime rate and disincentive respectively.

\textbf{Case (i) $\overline{\margCostCrime}_1 + \epsilon \le \overline{\margCostCrime}_2$}\\
First, because $\overline{\margCostCrime}_1 < \overline{\margCostCrime}_2$, $\margCostCrime^\margCostCrime_1 = \margCostCrime^\margCostCrime_2 = \overline{\margCostCrime}_1$. Now, as for $\margCostCrime^\crimeRate_g$, it depends on whether $\outOptCdf_2(\overline{\margCostCrime}_2) \le \outOptCdf_1(\overline{\margCostCrime}_1)$. Consider when $\outOptCdf_2(\overline{\margCostCrime}_2) \le \outOptCdf_1(\overline{\margCostCrime}_1)$. Then, we must have $\margCostCrime^\crimeRate_1 = \overline{\margCostCrime}_1$, and $\margCostCrime^\crimeRate_2$ should be such that $\outOptCdf_2(\margCostCrime^\crimeRate_2) = \outOptCdf_1(\overline{\margCostCrime}_1)$. Because $\outOptCdf_2$ stochastically dominates $\outOptCdf_1$, we have that $\margCostCrime^\margCostCrime_1=\margCostCrime^\margCostCrime_2 < \margCostCrime^\crimeRate_2$. By the monotonicity of $\outOptCdf_2$, it must be that $\outOptCdf_2(\margCostCrime^\margCostCrime_2) > \outOptCdf_2(\margCostCrime^\crimeRate_2)$.
Therefore,
\[
\numPeople_1\outOptCdf_1(\margCostCrime^\margCostCrime_1) + \numPeople_2\outOptCdf_2(\margCostCrime^\margCostCrime_2) > \numPeople_1\outOptCdf_1(\margCostCrime^\crimeRate_1) + \numPeople_2 \outOptCdf_2(\margCostCrime^\crimeRate_2).
\]

Now, consider when $\outOptCdf_2(\overline{\margCostCrime}_2) > \outOptCdf_1(\overline{\margCostCrime}_1)$. Then, we must have $\margCostCrime^\crimeRate_2 = \overline{\margCostCrime}_2$, and $\margCostCrime^\crimeRate_1$ should be such that $\outOptCdf_1(\margCostCrime^\crimeRate_1) = \outOptCdf_2(\overline{\margCostCrime}_2)$. Compare how each group's crime rate changes as we go from equalizing crime rate to equalizing disincentive. As for group $1$, it goes from $\outOptCdf_2(\overline{\margCostCrime}_2)$ to $\outOptCdf_1(\overline{\margCostCrime}_1)$. As for group $2$, it goes from $\outOptCdf_2(\overline{\margCostCrime}_2)$ to $\outOptCdf_2(\overline{\margCostCrime}_1)$.  Therefore, total change in crime rate by going from equalizing crime rate to equalizing disincentive is at most 0:

\begin{align*}
0&= \numPeople_1\left(\outOptCdf_2(\overline{\margCostCrime}_1 + \epsilon) - \outOptCdf_1(\overline{\margCostCrime}_1)\right) + \numPeople_2 \left(\outOptCdf_2(\overline{\margCostCrime}_1 + \epsilon) - \outOptCdf_2(\overline{\margCostCrime}_1)\right)\\
&\ge \numPeople_1\left(\outOptCdf_2(\overline{\margCostCrime}_2) - \outOptCdf_1(\overline{\margCostCrime}_1)\right) + \numPeople_2 \left(\outOptCdf_2(\overline{\margCostCrime}_2) - \outOptCdf_2(\overline{\margCostCrime}_1)\right)
\end{align*}

Therefore, we have that equalizing crime rates is better than equalizing disincentives.

\textbf{Case (ii) $\overline{\margCostCrime}_1 < \overline{\margCostCrime}_2 <  \overline{\margCostCrime}_1 + \epsilon$}\\
In this case, we know that $\outOptCdf_2(\overline{\margCostCrime}_2) > \outOptCdf_1(\overline{\margCostCrime}_1)$. For the optimal disincentive-equalizing policies, we have that $\margCostCrime^\margCostCrime_1 = \margCostCrime^\margCostCrime_2 = \overline{\margCostCrime}_1$. As for crime-equalizing policy, we have $\margCostCrime^\crimeRate_2 = \overline{\margCostCrime}_2$, and $\margCostCrime^\crimeRate_1$ is chosen such that $\outOptCdf_1(\margCostCrime^\crimeRate_1) = \outOptCdf_2(\overline{\margCostCrime}_2)$. Now, compare how each group's crime rate as we goes from equalizing crime rates to equalizing disincentives. As for group $2$, it goes from $\outOptCdf_2(\overline{\margCostCrime}_2)$ to $\outOptCdf_2(\overline{\margCostCrime}_1)$. As for group $1$, it goes from $\outOptCdf_2(\overline{\margCostCrime}_2)$ to $\outOptCdf_1(\overline{\margCostCrime}_1)$. Therefore, total change in crime rate by going from equalizing crime rate to equalizing disincentive is at most 0, as

\begin{align*}
&\numPeople_2 \left(\outOptCdf_2(\overline{\margCostCrime}_2) - \outOptCdf_2(\overline{\margCostCrime}_1)\right) + \numPeople_1\left(\outOptCdf_2(\overline{\margCostCrime}_2) - \outOptCdf_1(\overline{\margCostCrime}_1)\right) \\
&> \numPeople_2 \left(\outOptCdf_2(\overline{\margCostCrime}_1 + \epsilon) - \outOptCdf_2(\overline{\margCostCrime}_1)\right) + \numPeople_1\left(\outOptCdf_2(\overline{\margCostCrime}_1 + \epsilon) - \outOptCdf_1(\overline{\margCostCrime}_1)\right) \\
&= 0
\end{align*}

Therefore, equalizing disincentives is better than equalizing crime rates in this case.

\textbf{Case (iii) $\overline{\margCostCrime}_1 =\overline{\margCostCrime}_2$}\\
First, $\margCostCrime^\margCostCrime_1 = \margCostCrime^\margCostCrime_2 = \overline{\margCostCrime}_1 = \overline{\margCostCrime}_2$.
As for equalizing crime rates $\margCostCrime^\crimeRate_g$,  $\margCostCrime^\crimeRate_2 = \overline{\margCostCrime}_2$, and $\margCostCrime^\crimeRate_1$ is chosen such that $\outOptCdf_1(\margCostCrime^\crimeRate_1) = \outOptCdf_2(\margCostCrime^\crimeRate_2) > \outOptCdf_1(\overline{\margCostCrime}_1)$.
Therefore, we have
\[
\numPeople_1\outOptCdf_1(\margCostCrime^\margCostCrime_1) + \numPeople_2\outOptCdf_2(\margCostCrime^\margCostCrime_2) < \numPeople_1\outOptCdf_1(\margCostCrime^\crimeRate_1) + \numPeople_2 \outOptCdf_2(\margCostCrime^\crimeRate_2),
\]
meaning equalizing disincentives is better than equalizing crime rates.

\textbf{Case (iv) $\overline{\margCostCrime}_1 > \overline{\margCostCrime}_2$}\\
First, $\margCostCrime^\margCostCrime_1 = \margCostCrime^\margCostCrime_2 = \overline{\margCostCrime}_2$.
As for equalizing crime rates $\margCostCrime^\crimeRate_g$,  $\margCostCrime^\crimeRate_2 = \overline{\margCostCrime}_2$, and $\margCostCrime^\crimeRate_1$ is chosen such that $\outOptCdf_1(\margCostCrime^\crimeRate_1) = \outOptCdf_2(\margCostCrime^\crimeRate_2) > \outOptCdf_1(\overline{\margCostCrime}_1)$.
Therefore, we have
\[
\numPeople_1\outOptCdf_1(\margCostCrime^\margCostCrime_1) + \numPeople_2\outOptCdf_2(\margCostCrime^\margCostCrime_2) < \numPeople_1\outOptCdf_1(\margCostCrime^\crimeRate_1) + \numPeople_2 \outOptCdf_2(\margCostCrime^\crimeRate_2),
\]
meaning equalizing disincentives is better than equalizing crime rates.
\end{proof}

\eqBaseRatesOpt*
\begin{proof}
$(\overline{\margCostCrime}_1,\overline{\margCostCrime}_2)$ is the most optimal policy, and they equalize the crime rates $\outOptCdf_1( \overline{\margCostCrime}_1) = \outOptCdf_2(  \overline{\margCostCrime}_2)$. However, it is not guaranteed that the false positive/negative rates or disincentives will be equalized. For instance, if $\overline{\margCostCrime}_1 \neq \overline{\margCostCrime}_2$, then the disincentives are not equalized. And as $\margCostCrime_g = (1-\FNR_g) - \FPR_g$, false positive/negative rates won't be equalized in general.
\end{proof}

\end{document}